\documentclass[a4paper,10pt]{article}
\usepackage[centertags]{amsmath}
\usepackage{amsmath}
\usepackage[dutch,british]{babel}
\usepackage{amsfonts}
\usepackage{amssymb}
\usepackage{amsthm}
\usepackage{newlfont}
\usepackage{color}

 \usepackage{bbm}

\usepackage{stmaryrd}
\usepackage{mathrsfs}
\usepackage{pstricks,pst-node}
\usepackage{palatino}  
\usepackage{fancyhdr}
\usepackage{graphicx}
\usepackage{fancybox}
\usepackage{multibox}
\usepackage{geometry}
 \usepackage{setspace}
 \usepackage{psfrag}

\theoremstyle{plain}
  \newtheorem{theorem}{Theorem}[section]
  
  \newtheorem{proposition}[theorem]{Proposition}
  \newtheorem{lemma}[theorem]{Lemma}
  \newtheorem{remark}[theorem]{Remark}
\theoremstyle{definition}
  \newtheorem{definition}{Definition}[section]
  \newtheorem{assumption}[theorem]{Assumption}

\theoremstyle{remark}

\numberwithin{equation}{section}
\numberwithin{figure}{section}

\DeclareMathOperator{\Tr}{Tr}

\newcommand\otimesal{\mathop{\hbox{\raise 1.6 ex
  \hbox{$\scriptscriptstyle\mathrm{al}$}
\kern -0.92 em \hbox{$\otimes$}}}}
\newcommand\oplusal{\mathop{\hbox{\raise 1.6 ex
  \hbox{$\scriptscriptstyle\mathrm{al}$}
\kern -0.92 em \hbox{$\oplus$}}}}
\newcommand\Gammal{\hbox{\raise 1.7 ex
\hbox{$\scriptscriptstyle\mathrm{al}$}\kern -0.50 em $\Gamma$}}
\renewcommand\i{\mathrm{i}}

\let\al=\alpha \let\be=\beta \let\de=\delta \let\ep=\epsilon

  \let\ga=\gamma 
\let\ka=\kappa \let\la=\lambda \let\om=\omega 
\let\si=\sigma

 \let\Ga=\Gamma \let\La=\Lambda \let\Om=\Omega
  \let\Si=\Sigma

\newcommand{\caH}{{\mathcal H}}
\newcommand{\caI}{{\mathcal I}}
\newcommand{\caJ}{{\mathcal J}}

\newcommand{\caL}{{\mathcal L}}

\newcommand{\caP}{{\mathcal P}}

\newcommand{\caR}{{\mathcal R}}

\newcommand{\caT}{{\mathcal T}}
\newcommand{\caU}{{\mathcal U}}
\newcommand{\caV}{{\mathcal V}}
\newcommand{\caW}{{\mathcal W}}

\newcommand{\caZ}{{\mathcal Z}}

\newcommand{\bbC}{{\mathbb C}}

\newcommand{\bbE}{{\mathbb E}}

\newcommand{\bbN}{{\mathbb N}}

\newcommand{\bbR}{{\mathbb R}}
\newcommand{\bbS}{{\mathbb S}}
\newcommand{\bbT}{{\mathbb T}}

\newcommand{\bbZ}{{\mathbb Z}}

\newcommand{\opunit}{\text{1}\kern-0.22em\text{l}}

\newcommand{\frh}{{\mathfrak h}}

\newcommand{\scrB}{{\mathscr B}}
\newcommand{\scrC}{{\mathscr C}}

\newcommand{\scrE}{{\mathscr E}}

\newcommand{\scrH}{{\mathscr H}}

\newcommand{\pair}[1]{\langle{#1}\rangle}

\newcommand{\e}{{\mathrm e}}

\renewcommand{\d}{{\mathrm d}}

\renewcommand{\sp}{\mathrm{sp}}

\newcommand{\beq}{ \begin{equation} }
\newcommand{\eeq}{ \end{equation} }

\newcommand{\baq}{ \begin{eqnarray} }
\newcommand{\eaq}{ \end{eqnarray} }
\newcommand{\bet}{ \begin{theorem} }
\newcommand{\eet}{ \end{theorem} }

 \newcounter{smallarabics}
\newenvironment{arabicenumerate}
{\begin{list}{{\normalfont\textrm{\arabic{smallarabics})}}}
  {\usecounter{smallarabics}\setlength{\itemindent}{0cm}
  \setlength{\leftmargin}{5ex}\setlength{\labelwidth}{4ex}
  \setlength{\topsep}{0.75\parsep}\setlength{\partopsep}{0ex}
   \setlength{\itemsep}{0ex}}}
{\end{list}}

\newcounter{smallroman}

\newcommand{\ben}{\begin{arabicenumerate}}
\newcommand{\een}{\end{arabicenumerate}}

\newcommand{\sfock}{\Ga_{\mathrm{s}}}

\newcommand{\Symm}{\mathrm{Sym}}
\newcommand{\sys}{{\mathrm S}}
\newcommand{\res}{{\mathrm R}}

\newcommand{\ad}{\mathrm{ad}}
\newcommand{\adjoint}{\ad}

\newcommand{\dsi}{\si}

\newcommand{\norm}{ \|}
\newcommand{\str}{ |}
\newcommand{\lakl}{\lambda^{-2} }

\newcommand{\lat}{ \bbZ^d }
\newcommand{\tor}{ {\bbT^d}  }
\newcommand{\initialres}{\rho_\res^\be}

\newcommand{\links}{L}
\newcommand{\rechts}{R}

\newcommand{\bosondispersion}{\om}
\newcommand{\systemdispersion}{\varepsilon}

\begin{document}

\begin{center}
\large{ \bf{Quantum Brownian Motion in a Simple Model System}} \\
\vspace{20pt} \normalsize

{\bf   W.  De Roeck\footnote{Postdoctoral Fellow FWO-Flanders at  K.U.Leuven, Belgium, 
email: {\tt
 wojciech.deroeck@fys.kuleuven.be}}  }\\

\vspace{10pt} 
{\bf   Institute for Theoretical Physics \\
K.U.Leuven \\
B3001 Heverlee, Belgium} \\
\vspace{5pt} 
{\bf   Institute for Theoretical Physics \\
ETH Z\"urich \\
CH-8093 Z\"urich, Switzerland} \\

\vspace{20pt}

{\bf   J. Fr\"ohlich   }\\
\vspace{10pt} 
{\bf   Institute for Theoretical Physics \\
ETH Z\"urich \\
CH-8093 Z\"urich, Switzerland}
\vspace{20pt} 

{\bf A. Pizzo  }

\vspace{10pt} 
{\bf  Department of Mathematics \\
University of California at Davis \\
Davis CA 95616, USA
}

\vspace{15pt} \normalsize

\end{center}

\vspace{20pt} \footnotesize \noindent {\bf Abstract: }
We consider a quantum particle coupled  (with strength $\la$) to a spatial array of independent non-interacting reservoirs in thermal states (heat baths). Under the assumption that the reservoir correlations decay exponentially in time, we prove that the motion  of the particle  is diffusive at large times for small, but finite $\la$.  Our proof relies on an expansion around the kinetic scaling limit ($\la \searrow 0$, while time and space scale as $\la^{-2}$) in which the particle satisfies a Boltzmann equation. We also show an equipartition theorem:  the distribution of the kinetic energy of the particle tends to a Maxwell-Boltzmann distribution, up to a correction of $O(\la^2)$.

\vspace{5pt} \footnotesize \noindent {\bf KEY WORDS:}  diffusion, kinetic limit, quantum brownian motion  \vspace{20pt}
\normalsize
\section{Introduction}\label{sec: intro}

\subsection{Diffusion }

Diffusion and Brownian motion are among  the most fundamental phenomena described by transport theory.  They refer to the apparent random motion of a particle or, for that matter, any degree of freedom,  interacting with many other, mutually independent degrees of freedom in a thermal state.  The interactions  produce an erratic macroscopic motion that we perceive as diffusive or as Brownian motion.  
From a mathematical point of view, we may attempt to understand diffusive motion by invoking a central limit theorem:  $N$  interactions produce an effect $\delta x$, which is given by $\delta x \sim \sqrt{N}$. Since the number of interactions is proportional to the time lapse $\delta t$, we can write
$
(\delta x )^2 = D \delta t
$, where the proportionality constant $D$ is called the diffusion constant.  Via the Einstein relation, the diffusion constant determines  quantities such as the thermal or electric conductivity.

The model of a particle (quantum or classical)  coupled to a thermal reservoir of free particles is a natural starting point for an analysis of diffusion. We assume  the particle to be   quantum mechanical.
 By $\langle \cdot \rangle_\be$ we denote the expectation value in a state where the reservoir has an inverse temperature $\be<\infty$.  Then
\beq
\langle  (\delta x )^2 \rangle_\be =  \langle  (x(t)-x(0) )^2 \rangle_\be = \int_0^t \d t_1 \int_0^t \d t_2   \langle  \dot{x}(t_1)  \dot{x}(t_2) \rangle_\be,
\eeq
where $x(t)$ is the position of the particle at time $t$ and $\dot{x}(t)= \i [H, x(t)]$, where $H
$ is the Hamiltonian of the system, is the velocity. 
We expect that, because of interactions with the reservoir,  $\dot{x}(t_1) $ and $ \dot{x}(t_2) $ become de-correlated rapidly, as $\str t_1-t_2 \str$ grows.   Thus,  the quantity $\str  \langle  \dot{x}(t_1)  \dot{x}(t_2) \rangle_\be \str$ is expected to be integrable  in the variable $t_2-t_1$. Combining this with isotropy, i.e.,  $ \langle  \dot{x}(t)  \rangle_\be \rightarrow 0$ rapidly, as $t \rightarrow \infty$, for $\be <\infty$,  one concludes that,  asymptotically as $t$ tends to $\infty$,
\beq \label{eq: D as variance}
 \langle  (x(t)-x(0) )^2 \rangle_\be = D(\be) \str t \str, 
\eeq
for some positive, finite constant $0 < D(\be) < \infty$, given by
\beq\label{def Dbe}
D(\be)  =  \int_{\bbR} \d t  \langle  \dot{x}(0)  \dot{x}(t) \rangle_\be.
\eeq
Because of  the equipartition theorem, one expects  that 
\beq \label{square velocity}
 \langle  \dot{x}(t)^2   \rangle_\be \rightarrow  \int_{\bbR^d} \d v  \,  v^2 \e^{-\be E(v)},  \qquad  \textrm{as} \, \, t \nearrow \infty,
\eeq
where  $E(v)$ is the kinetic energy of a particle with velocity $v$.   Obviously, \eqref{square velocity} is strictly positive for finite $\be$. Likewise, we expect that $D(\be)$ is strictly positive, for $\be < \infty$.

Equations \eqref{eq: D as variance} and \eqref{square velocity} suggest that, at very large times, the motion of a particle interacting with a reservoir or heat bath at strictly positive temperature has universal features: The mean value of its speed is strictly positive and finite, and its mean displacement is proportional to the square root of time. In contrast, at zero temperature ($\beta=\infty$), the nature of the particle's motion depends on properties of the reservoir and the dispersion law, $\systemdispersion(k)$, of the particle; ($k \in \bbR^d$ is its momentum). If, for a particle momentum $k$, 
\beq\label{eq: no emission}
\systemdispersion(k-q)+\bosondispersion(q) > \systemdispersion(k) ,Ê\qquad \textrm{for all} \, q \neq 0,
\eeq
where $\om(q)$ is the dispersion law of a mode (particle) of the reservoir with momentum $q$, then the particle cannot lower its energy and reduce its speed by exciting a reservoir mode, i.e., by spontaneously emitting a reservoir particle. Its motion will therefore be ballistic. The only effect of the reservoir is a renormalization of the effective mass (the dispersion law $\systemdispersion$) of the particle. If, however, \eqref{eq: no emission} is not satisfied, then the particle can excite reservoir modes (emit reservoir particles).  This process reduces its kinetic energy and speed, i.e., it leads to friction. Friction takes place at all momenta $k$ if, e.g., $\om(q) \propto \str  q \str^2$ (reservoir particles are non-relativistic). If $\om(q) =c \str q \str$, i.e., the reservoir particles are low-energetic phonons or photons, friction only takes place at momenta $k$ of the particle where $\str   \nabla \systemdispersion (k)\str >c$. The radiation corresponding to the reservoir particles emitted in the process of friction is called Cerenkov radiation.

Despite the importance of diffusion and its conceptual simplicity, there has, so far, not existed any rigorous proof that it occurs in a model as described above. 
In the present paper, we establish diffusion for models where the particle is coupled to a spatial array of independent heat baths.

\subsection{Informal description of the model and  main results}

We consider a quantum particle hopping on the lattice $\bbZ^d$. With  each lattice point, we associate   an independent thermal  reservoir consisting of  a free bosonic quantum field describing phonons or photons  at temperature $\be^{-1}$. (In this section, we present a description of the system appropriate at zero temperature; it is formal when $\be <\infty$.)   The total Hilbert space, $\scrH$, of the coupled system is a tensor product of the system space, $\scrH_\sys$, with a reservoir space, $\scrH_\res$, which is a (separable) subspace of the infinite tensor product of reservoir spaces  $\scrH_{\res_x}, x \in \bbZ^d$, at all sites. Thus 
\beq
\scrH:= \scrH_\sys \otimes  \scrH_{\res}.  
\eeq
The system space $\scrH_\sys$ is given by $l^2(\bbZ^d)$, and the particle Hamiltonian is given by the finite-difference Laplacian $\Delta$. Each reservoir is described by a boson field;  creation and annihilation operators creating/annihilating bosons with momentum $q \in \bbR^d$ at site $x$ are written as $ a^{*}_x(q), a_x(q)$ respectively, and  satisfy the canonical commutation relations
\beq
 [a^{\#}_x(q) , a^{\#}_{x'}(q')] =0, \qquad  [a_x(q) , a^{*}_{x'}(q')] = \delta_{x,x'}\delta(q-q'),
\eeq
where $a^{\#}$ stands for either $a$ or $a^{*}$.

The total Hamiltonian of the system is taken  to be
\beq  \label{def: hamiltonian1}
H_\la: = - \Delta    +  \sum_{x \in \lat}  \int_{\bbR^d} \d q   \bosondispersion(q)     a^{*}_x(q) a_x(q)   + \la    \sum_{x \in \lat}  \int_{\bbR^d} \d q\,     \str x \rangle \langle x \str \otimes \left\{    \phi(q) a^{*}_x(q)   +\mathrm{h.c.} \right\},
\eeq 
where $\phi(q)$ is a form factor and $\la \in \bbR$ is the coupling strength. We are writing $\Delta$ instead of $\Delta \otimes 1$ and $a_x(q)$ instead of $1 \otimes a_x(q)$ 

  The independence of the reservoirs has far-reaching consequences. 
Consider the lattice translation $\caT_z, z \in \lat$, acting  on operators on $\scrH$  by  
\baq
\caT_z(  \str x \rangle \langle y \str  ) & := &    \str x+z \rangle \langle y+z \str     \\
\caT_z( a^{\#}_x(q)  ) & := &   a^{\#}_{x+z}(q).
\eaq
It is easily seen that
\beq
\caT_z (H_\la)= H_\la.
\eeq
Notice that this transformation does not involve the momentum coordinates $q$ inside the reservoirs.
It is  the existence of this translation symmetry that  allows us to obtain  results on diffusion without very hard work.   Assume we had started from a model with only one reservoir, with Hamiltonian given by
\beq \label{def: hamiltonian2}
H_\la := - \Delta    +   \int_{\bbR^d} \d q   \om(q)     a^{*}(q) a(q)   + \la    \sum_{x \in \lat}  \int_{\bbR^d} \d q  \,   \str x \rangle \langle x \str  \otimes \left\{  a^{*}(q)   \phi(q) \e^{-\i (x,q)}   +\mathrm{h.c.} \right\},
\eeq 
where, now, the operators $a(q),a^*(q)$ do not carry an index $x$ and  $(\cdot,\cdot)$ is the scalar product on $\bbC^d$.  This model still  exhibits translation symmetry, but this symmetry maps  $ a^*(q) \rightarrow a^*(q) \e^{-\i (z,q)}   , a(q)  \rightarrow  \e^{\i (z, q) }a(q)  $, which is the reason for the factor $\e^{-\i (x,q)} $  in the interaction  Hamiltonian of \eqref{def: hamiltonian2} and leads to  bad decay properties of the reservoir correlation functions.

The initial state for the reservoirs is chosen to be $\initialres := \otimes_{x \in \lat} \rho_{\res_x}^\be$,  where each $\rho_{\res_x}^\be$ is an equilibrium state at inverse temperature $\be$ for the reservoir at site $x$. For mathematical details on the construction of infinite reservoirs, see \cite{derezinski1,bratellirobinson,arakiwoods}.
 In Lemma \ref{lem: definition dynamics}, we   define the reduced Heisenberg-picture dynamics (i.e.,  the particle dynamics obtained by tracing out the reservoir degrees of freedom)  
\beq
S \mapsto \caZ_{t}^{\la,*} (S) := \initialres \left[    \e^{i t H_\la}  (S  \otimes 1)   \e^{-i t H_\la}   \right], \qquad S \in \scrB(\scrH_\sys)
\eeq
Placing the particle initially at site $0$, that is, in the vector $\str 0 \rangle$, we study the distribution function
\beq
\mu_t^\la(x):=  \langle 0 \str  \caZ_{t}^{\la,*} (\str x \rangle \langle x \str )  \str 0 \rangle.
\eeq
 The quantity $\mu_t^\la(x) \geq 0$ is the probability to find the particle at site $x$  at time $t$.  One  easily  checks that 
\beq
\sum_{x \in \lat}  \mu_t^\la(x) =1, 
\eeq
and hence it is justified to think of $\mu_t^\la(\cdot)$ as a probability density on $\lat$.  By diffusion, we mean that, for large $t$, 
\beq\label{def: diffusion1}
\mu_t^\la(x)    \sim  \left( \frac{1}{2\pi t}\right)^{d/2} (\mathrm{det}D_\la)^{-1/2} \exp \{ -  \left(  \frac{x  }{\sqrt{t}} , D^{-1}_\la  \frac{x  }{\sqrt{t}}  \right) \},
\eeq
where $D_\la \equiv D_\la(\be)$ is a  positive-definite matrix with the interpretation of a diffusion tensor. (Actually, if the particle Hamiltonian is given by $-\Delta$ (as in this section), the tensor $D_\la$ is isotropic and hence a scalar). 
We now move towards quantifying \eqref{def: diffusion1}.  Let us fix a time $t$. Since $\mu_t^\la(x)$ is a probability measure, one can think of $x_t$ as a random variable such that 
\beq\label{def: random variable}
\mathrm{Prob}_\la(x_t=x): = \mu_t^\la(x).
\eeq
The claim that the random variable $\frac{x_t}{\sqrt{t}}$ converges in distribution, as $t \nearrow \infty$, to a Gaussian random variable with mean $0$ and variance $D^{-1}_\la$ is called a Central Limit Theorem (CLT). 
It is equivalent to pointwise convergence of the characteristic function, i.e., 
\beq \label{eq: conv characteristic function}
 \sum_{x \in \lat} \e^{ -\frac{\i }{\sqrt{t}} (x,q) }  \mu_t^\la(x)   \quad \mathop{\longrightarrow}\limits_{t \uparrow \infty} \quad  \e^{- \frac{1}{2}(q, D_\la q)}, \quad  \textrm{for all}  \, q \in \bbR^d,
\eeq
and it is this statement which is our main result, Theorem \ref{thm: diffusion}. 

 Let $X:= \sum_{x \in \lat}x\,  \str x \rangle \langle x \str$ be the position operator on the lattice and write $X_t:= \caZ_{t}^{\la,*}(X)$. Then a slightly stronger version of \eqref{eq: conv characteristic function} implies that 
 \beq \label{eq: convergence of second moment}
 \langle 0 \str \frac{ X_t}{t}  \str 0 \rangle  \quad \mathop{\longrightarrow}\limits_{t \uparrow \infty} \quad  0 , \qquad     \langle 0 \str \frac{ X^2_t}{t}  \str 0 \rangle  \quad \mathop{\longrightarrow}\limits_{t \uparrow \infty} \quad   D_{\la},
 \eeq
 and this will also follow from our results; see Remark \ref{rem: convergence of moments}.
 
 Our second result concerns the asymptotic expectation value of the kinetic energy of the particle. Let $E_t:=\caZ_{t}^{\la,*}(-\Delta)$ be the kinetic energy at time $t$. We prove that, for all bounded  functions $\theta$,
 \beq
  \langle 0 \str  \theta(E_t) \str 0 \rangle  \quad \mathop{\longrightarrow}\limits_{t \uparrow \infty} \quad    \frac{\int_{\tor} \d k \,  \theta(\systemdispersion(k))   \e^{-\be \systemdispersion(k)}   }{    \int_{\tor} \d k  \,  \e^{-\be \systemdispersion(k)}   } +O(\la^2), \qquad  \la \downarrow 0,
 \eeq
 where $\systemdispersion(k)=\sum_{j=1}^d (2-2\cos k^j) $ is the dispersion law of the particle.  This is stated in Theorem \ref{thm: stationary}.

\subsection{Related results}

In the physics literature, the model with Hamiltonian \eqref{def: hamiltonian2}  and with reservoir particles being phonons is referred to as the polaron model. We refer to \cite{spohnkineticreview,silviusparrisdebievre} and references therein for a discussion. The first rigorous result on this model at positive temperature is probably in  \cite{spohnelectronrandomimpurity} and the best result up to date is in \cite{erdos}; (see also Section \ref{sec: kinetic limit}).  

To describe some related results, we first introduce a different  model, which, however, will turn out to be closely related to ours. 

Assume that the quantum particle interacts with random time-dependent impurities. That is, let $V(x,t)$ be  a real-valued random variable, for $x \in \lat, t \in \bbR$, with mean zero
\beq
\bbE \left[V(x,t) \right]=0, \eeq
 satisfying the Gaussian property
\baq 
\bbE \left[V(x_{2n},t_{2n})   \ldots   V(x_1,t_1)  \right] &=& \sum_{\mathrm{pairings} \,\pi}\prod_{(r,s) \in \pi}    \bbE \left[V(x_s,t_s)     V(x_r,t_r)  \right] \label{guassianity1}\\[2mm]
\bbE \left[V(x_{2n+1},t_{2n+1})   \ldots   V(x_1,t_1)  \right] &=& 0 \label{guassianity2},
  \eaq
where a pairing $\pi$ is a partition of $\{1,\ldots,2n \}$ into $n$ pairs and the product is over these pairs $(r,s)$. In addition, we assume that the correlation functions are invariant under translations  in time and space,
\beq
\bbE \left[V(x,t)     V(x',t')  \right] = \bbE \left[V(x-x',t-t')     V(0,0)  \right] . \eeq
A time-dependent Hamiltonian is given by
\beq
H_{\la}(t) := - \Delta +   \la \sum_{x \in \lat}      V(x,t)   \str x\rangle\langle x \str,
\eeq
and the dynamics $U_t^\la$ is defined (almost surely) by  
\beq
 \frac{\d}{\d t}U^\la_t= -\i H_\la(t) U^\la_t , \qquad  U^\la_0=1.
\eeq
One can check that if we choose
\beq\label{eq: correspondance}
\bbE    \left[   V(x,t) V(0,0)   \right]    :=  \delta_{x,0} \int_{\bbR^d} \d q    \str \phi(q)\str^2  \left(   \initialres\left[  a_0(q) a_0^{*}(q)   \right]  \e^{-\i t \bosondispersion(q)}+ \initialres\left[   a_0^{*}(q)  a_0(q) \right]  \e^{\i t \bosondispersion(q)}  \right),
\eeq
(the RHS will be motivated  in Section \ref{sec: full reservoir}),
then we have that
\beq
\bbE \left[ U^\la_t  \rho (U^\la_t)^* \right]   =     \caZ_t^{\la}(\rho).
\eeq
The reason for this equivalence is that both models share a ``quasi-free", or, `Gaussian property''. (In the Hamiltonian model, this is a consequence of the fact that the free reservoir Hamiltonian is quadratic in the creation and annihilation operators).
Of course, it is not clear that the definition \eqref{eq: correspondance} makes sense. For example, the RHS could have an imaginary part, whereas the LHS is real.
However, upon inspection of our proof, it becomes clear that whenever 
\beq
\big\str \bbE    \left[   V(x',t') V(x,t)   \right]  \big\str    \leq  \delta_{x,x'}  c \e^{-g_\res \str t-t'\str}, \qquad c<\infty, g_\res <\infty,
\eeq
then our proof (which assumes the same bound for the RHS of  \eqref{eq: correspondance}) carries over, and we can establish diffusive behavior for the averaged dynamics $\bbE \left[ U^\la_t  \rho (U^\la_t)^* \right] $.
In fact, the locality in space (expressed by $  \delta_{x,x'} $) is not crucial, at all, but we do  not pursue this generalization here. 

The case 
$
\bbE    \left[   V(x',t') V(x,t)   \right]      =  \delta_{x,x'}  \delta(t-t')
$
has been treated in \cite{ovchinnikoverikhman}. While we were completing this paper, a preprint \cite{kangschenker} appeared where diffusion is proven under the assumption that $V(x,t)$ is an exponentially ergodic Markov process (not necessarily Gaussian) for each $x$. Preliminary results were obtained in \cite{pilletmarkovianpotential} and \cite{tcheremchantsev}. 
One of the ultimate goals of these projects is to treat the case where $V(x,t)=V(x)$ is time-independent and $d=3$, i.e., $\bbE    \left[   V(x') V(x)   \right] = \delta_{x,x'} $. This is the well-known  Anderson model.

Models in which the particle is coupled to a thermal reservoir are expected to  be easier than the Anderson model, mainly because one expects that diffusion persists for large values of the coupling constant $\la$, whereas  the Anderson model has a phase transition, and the particle gets localized at large values of $\str\la\str$. 

However, even for a particle coupled to a thermal reservoir in $d=3$,  our techniques  fail, since this model would essentially correspond to one with
$
\bbE    \left[   V(x',t') V(x,t)   \right]      \sim   \frac{1}{\str x -x' \str}   \chi[ \str x -x' \str \geq c \str t -t' \str ],
$
(for  reservoir particles with dispersion relation $\bosondispersion(q)=c\str q\str$).

There are however results that establish diffusive behavior up to times of order $\la^{-(2+\delta)}$, for some $\delta>0$, even for the Anderson model, see \cite{erdossalmhoferyaunonrecollision,erdossalmhoferyaurecollision} (a resulting lower bound for the localization length is proven in \cite{chenlocalization}).  In fact, our technique employs results of the type proven in these references as an ingredient of the proof; see Section \ref{sec: kinetic limit}. 

We might add that we expect that the model treated in the present paper can also be analyzed using operator-theoretic techniques introduced for the study of return to equilibrium in open quantum systems, see e.g.\ \cite{jaksicpillet2,bachfrohlichreturn} , and we are currently working on such a formulation.  The technique used in the present paper is largely based on \cite{deroeck}.

\subsection{Outline}
In Section \ref{sec: model}, we introduce our model,  making precise the description in the introduction. Then, in Section \ref{sec: result}, we state our assumptions and  main results with as few divagations as possible.  Section \ref{sec: discussion} contains the main ideas of the paper and the plan of the proof.  The technical parts of the proof are  postponed to  Section \ref{sec: dyson and proof of polymer}, which contains the proof of Theorem \ref{thm: polymer model}, and  Section \ref{sec: proof of main technical}, where one finds the proof of Theorem \ref{thm: main technical improved} .

\subsection{Acknowledgements}

W.D.R. thanks J. Bricmont for helpful discussions  on an early version of this model.


\section{Model} \label{sec: model}

\subsection{Conventions and notation} \label{sec: conventions}

Given a Hilbert space $\scrE$, we use the standard notation 
\beq \scrB_p(\scrE):= \left\{  S \in \scrB(\scrE), \Tr\left[(S^*S)^{p/2}\right] < \infty  \right\}  ,\qquad   1 \leq p \leq \infty,  \eeq
with $\scrB_\infty(\scrE)\equiv \scrB(\scrE)$ the bounded operators on $\scrE$,
and
\beq
\norm S \norm_p := \left(\Tr\left[(S^*S)^{p/2}\right]\right)^{1/p}, \qquad  \norm S \norm:=\norm S \norm_{\infty}.
\eeq

For bounded  operators acting on $\scrB_p(\scrE)$, i.e. elements of $\scrB(\scrB_p(\scrE))$, we use in general  the calligraphic font: $\caV,\caW,\caT,\ldots$.
An operator $X \in \scrB(\scrE)$ determines an operator  $\adjoint(X) \in \scrB(\scrB_p(\scrE))$ by 
\beq
\adjoint(X) S: =[X,S]= XS-SX, \qquad   S \in \scrB_p(\scrE).
\eeq
We will mainly use the case $p=2$.
 The norm of  operators in $\scrB(\scrB_2(\scrE))$ is defined by \beq\label{def: norm on operators}
\norm \caW \norm := \sup_{S \in \scrB_2(\scrE)}   \frac{\norm \caW(S) \norm_2}{\norm S\norm_2}.
\eeq

For vectors $\ka \in \bbC^d$, we let $\Re \ka, \Im \ka$ denote the vectors $(\Re \ka^1, \ldots, \Re \ka^d)$ and  $(\Im \ka^1, \ldots, \Im \ka^d)$, respectively. The scalar product on $\bbC^d$ is written as 
$(\cdot, \cdot)$ and the norm  as $\str \ka \str := \sqrt{(\ka, \ka)}$. 
The scalar product on  an infinite-dimensional Hilbert space $\scrE$  is written as $\langle \cdot, \cdot \rangle$, or, occasionally, as $\langle \cdot, \cdot \rangle_{\scrE}$. All scalar products are defined to be  linear in the second argument and anti-linear in the first one.

We write $\sfock(\scrE)$ for the symmetric (bosonic) Fock space over the Hilbert space $\scrE$ and we refer to  \cite{derezinski1} for definitions and discussion. 
If $\om$ is a self-adjoint operator on $\scrE$, then its (self-adjoint) second quantization, $\d \sfock (\om)$, is defined by 
\beq
\d \sfock (\om)  \Symm (\phi_1 \otimes \ldots \otimes \phi_n)   := \sum_{i=1}^n    \Symm (\phi_1 \otimes\ldots \otimes  \om \phi_i \otimes \ldots \otimes \phi_n),  
\eeq
where $ \Symm$ projects on the symmetric subspace and $\phi_1, \ldots, \phi_n \in \scrE$.

\subsection{The particle}\label{sec: particle}

We set $\scrH_\sys=l^2(\bbZ^d)$  (the subscript $\sys$ refers to 'system', as is customary in system-reservoir models). 
We define the one-dimensional projector $1_x $ on $\scrH_\sys$ by 
\beq
(1_x f)(x'):= \delta_{x,x'} f(x') ,\qquad         x,x' \in \lat,  f \in l^2(\lat).
\eeq
We will often consider the space $\scrH_\sys$ in its dual representation, i.e.\ as $L^2(\bbT^d, \d k) $ where $\bbT^d$ is the $d$-dimensional torus, which is identified with $L^2([-\pi,\pi]^d)$.  We define the  `momentum' operator $P$ as multiplication by $k \in \bbT^d$, i.e.,
\beq
(P \theta)(k):=k \theta(k), \qquad    \theta \in L^2(\tor, \d k).
\eeq
Although $P$ is well-defined as a bounded operator, it does not have nice properties; e.g., it is not true that $[X^i,P^j]=\i \delta_{i,j}$. 
Throughout the paper, we only use operators  $f(P)$ where $f$ is periodic on $\bbR^d$ with period $2\pi$, i.e.\ a function on $\tor$.  We choose a periodic function $\systemdispersion$ to be the dispersion law of the system. Although this is not essential, we require $\systemdispersion$ to have inversion symmetry, i.e.,
  \beq
  \systemdispersion(k)=\systemdispersion(-k), \qquad  k \in \bbT^d.
  \eeq
 The Hamiltonian of our particle is  given by
 \beq
 H_\sys:= \systemdispersion(P).
 \eeq
 Our first assumption ensures that $H_\sys$ is sufficiently regular.
\begin{assumption}[Analyticity of system dynamics] \label{ass: analytic dispersion}
The function $\systemdispersion$, defined originally on $\bbT^d$, extends to an analytic function  in a strip of width $\delta_\systemdispersion>0$.  That is, when viewed as a periodic function on $\bbR^d$,  $\systemdispersion$ is analytic in $(\bbR+ \i [-\delta_\systemdispersion, \delta_\systemdispersion])^d $. 
Moreover, we assume that the function $\tor \ni k \mapsto
 ( \upsilon, \nabla\systemdispersion(k))$ does not vanish identically for any vector $\upsilon \in \bbR^d, \upsilon \neq 0$.
\end{assumption}
The most natural choice for $\systemdispersion$ satisfying Assumption \ref{ass: analytic dispersion} is $\systemdispersion(k)= \sum_{j=1}^d  (2-2\cos (k^j))$, which corresponds to $-H_\sys$ being the discrete Laplacian.

\subsection{The reservoirs} \label{sec: reservoir}

\subsubsection{Reservoir spaces}
We consider an array of independent reservoirs. With each site $x \in \lat$ we associate a one-particle Hilbert space
$\frh_{x}$ (one can imagine that  $\frh_{x}= L^2(\bbR^d)  $)  with a positive one-particle Hamiltonian $\bosondispersion_{x}$. 
The reservoir at $x$ is now described by the Fock space $\sfock (\frh_{x})$ with Hamiltonian $\d \sfock(\bosondispersion_x)$.
 The full reservoir space is
 \beq \scrH_\res :=  \sfock (\oplus_{x \in \lat}  \frh_{x}) \quad \textrm{with Hamiltonian} \quad H_\res :=  \sum_{x \in \lat} \d \sfock (\bosondispersion_{x}). \eeq  

We choose the different reservoir one-particle spaces to be isomorphic copies of a fixed space $\frh$ so that  $\varphi \in \frh_{x}$ is naturally identified with an element of $ \frh_{x'}$ that is also denoted by  $\varphi$ without further warning. Likewise, $\bosondispersion_{x}$ is naturally identified with $\bosondispersion_{x'}$. Hence, if no confusion is possible
 we  simply write  $\frh$ and $\bosondispersion$ to denote the (one-particle) one-site space and the Hamiltonian, respectively.

For $\varphi \in \frh$, the  operators $a^*_x(\varphi)/ a_x(\varphi)$ stand for the creation/annihilation operators on the Fock space $\sfock (\frh_{x})$.  By the embedding of $\frh_x$ into $\oplus_{y \in \lat}  \frh_{y}$, these creation/annihilation operators act on $\scrH_\res$ in a natural way. They satisfy the commutation relations
\beq
[   a_x(\varphi), a^*_{x'}(\varphi')  ]  = \delta_{x,x'} \langle \varphi, \varphi' \rangle_{\frh}, \qquad   [   a^{\#}_x(\varphi), a^{\#}_{x'}(\varphi')  ]  = 0
\eeq
where $a^{\#}$ stands for either $a^*$ or $a$.
\subsubsection{Interaction and initial reservoir state}\label{sec: full reservoir}

We pick  a `structure factor' $\phi \in\frh$ and we choose the interaction between the system and the reservoir at site $x$ to be given by
\beq
1_x  \otimes  \Psi_x(\phi),  \quad \textrm{where} \quad    \Psi_x(\phi)=   a_x(\phi)+ a_x^*(\phi) 
\eeq
So far, we have not made any assumptions concerning $\bosondispersion$ and $\phi$, but their form will be restricted by Assumption \ref{ass: exponential decay} in \eqref{eq: ass decay}.
The particle interacts with all reservoirs in a translation invariant way. Hence the total interaction Hamiltonian is given by
\beq
H_{\sys\res}:= \sum_{x \in \lat}   1_x  \otimes \Psi_x(\phi) \quad   \textrm{on} \quad 
\scrH_\sys \otimes  \scrH_\res.
\eeq

Next, we put the tools in place to describe the positive temperature  reservoirs. 
Let $\scrC$ be the $*$-algebra consisting of polynomials in $a_x(\varphi),a_{x'}^*(\varphi')$, with $\varphi,\varphi' \in \frh, x,x' \in \lat$. We introduce the positive operator $T_\beta= (\e^{\be \bosondispersion}-1)^{-1}$ on $\frh$; $\beta$ should be thought of as the inverse temperature.

We let  $\initialres $ be a quasi-free state defined on $\scrC$. It is fully specified\footnote{The reason why, in models like ours, it is enough to know the state on $\scrC$, has been explained in many places, e.g.\  \cite{arakiwoods,  bratellirobinson, jaksicpilletderezinski, froehlichmerkli}. } by 
\ben
\item{ Gauge-invariance 
\beq \initialres \left[  a^*_x(\varphi)
 \right]=
 \initialres \left[  a_x(\varphi) \right]= 0. \eeq}
 \item{Two-point correlation functions
  \beq \label{def: thermal state canonical}  \left( \begin{array}{cc} \initialres \left[ a^*_{x}(\varphi)
a_{x'}(\varphi') \right]
& \initialres \left[a^*_{x}(\varphi) a^*_{x'}(\varphi') \right]  \\
\initialres \left[ a_{x}(\varphi)a_{x'}(\varphi')\right]& \initialres \left[ a_{x}(\varphi) a^*_{x'}(\varphi')
\right]
\end{array} \right)  =   \delta_{x,x'} \left(\begin{array}{cc}
 \langle \varphi' | T_\beta \varphi \rangle & 0    \\
0 & \langle \varphi |(1+ T_\beta) \varphi') \rangle
\end{array}\right). \eeq
 }
\item{ Quasi-freeness, i.e.\ , the higher-point  correlation functions are expressed in terms of the two-point function by
\baq 
\initialres \left[a^{\#}_{x_1}(\varphi_1)  \ldots a^{\#}_{x_{2n}}(\varphi_{2n}) \right] & =&   \sum_{\mathrm{pairings}\, \pi} \prod_{(r,s) \in \pi}   \initialres \left[a^{\#}_{x_r}(\varphi_r) a^{\#}_{x_s}(\varphi_s) \right]   \label{eq: gaussian property1}\\[2mm]
\initialres \left[a^{\#}_{x_1}(\varphi_1)  \ldots a^{\#}_{2n+1}(\varphi_{2n+1}) \right]  &=&  0. \label{eq: gaussian property2}
\eaq
where a pairing $\pi$ is a partition of $\{1,\ldots,2n \}$ into $n$ pairs and the product is over these pairs $(r,s)$.
 }
\een

A quantity that will play an important role  in our analysis is the  on-site-reservoir correlation function defined by 
\baq
\hat\psi(t ) &:=&  \initialres  \left[    \Psi_x(\e^{\i t \bosondispersion}\phi)      \Psi_x(\phi)   \right]   \nonumber  \\
&=&   \langle \phi , T_\beta \e^{\i t \bosondispersion } \phi \rangle   +      \langle \phi , (1+T_\beta) \e^{- \i t \bosondispersion } \phi\rangle \label{def: correlation function} .
\eaq
It is useful to introduce $\psi$, the inverse Fourier transform of $\hat\psi$
\beq \label{def: effective structure factor}
 \hat \psi(t )  =  \frac{1}{\sqrt{2\pi}}\int_{\bbR} \d \xi  \,    \e^{\i \xi t }    \psi(\xi).
 \eeq
 As is explained in Appendix A, $\psi$   is the (squared norm of) the \emph{effective structure factor}. In particular, $\psi(\xi)\geq 0$.

The following assumption requires the  reservoir to have exponential decay of correlations.

\begin{assumption} \label{ass: exponential decay}
There is a  decay rate $ g_\res >0$ such that 
\beq \label{eq: ass decay}
 \sup_{t \in \bbR}  \left(\str \hat\psi(t)  \str   \e^{ g_\res \str t \str  }  \right) < \infty.
\eeq	
We assume that $\hat\psi \not\equiv 0$, or equivalently $\psi  \not\equiv 0$.
\end{assumption}
The  assumption that $\hat\psi \not\equiv 0$ ensures that the particle  interacts effectively with the fields describing the reservoirs. In Appendix A, we discuss examples of reservoirs that satisfy Assumption \ref{ass: exponential decay}, provided that $\beta < \infty$.

\subsection{The dynamics}

Consider the zero-temperature Hilbert space $\scrH_\sys \otimes \scrH_\res$. The  Hamiltonian (with coupling constant $\la)$ is formally defined by
\beq \label{def= Hlambda}
H_\la : = H_\sys + H_\res+ \la H_{\sys\res}.
\eeq
This operator generates the zero-temperature dynamics. However, we need to consider the dynamics at positive temperature. In particular, we must understand the reduced positive-temperature dynamics of the system $\sys$ after the reservoir degrees of freedom have been traced out.

By  a slight abuse of notation, we use $\initialres$ to denote the conditional expectation  from $\scrB(\scrH_\sys \otimes \scrC) $ to  $ \scrB(\scrH_\sys) $ given by 
\beq
 \initialres (S \otimes R):= S \initialres(R), \qquad  S \in \scrB(\scrH_\sys), R \in \scrC
\eeq
where $\initialres(R)$ is defined through (\ref{def: thermal state canonical}-\ref{eq: gaussian property2}).

Formally, the reduced dynamics in the Heisenberg picture is given by 
\beq \label{def: z}
\caZ_{t}^{\la,*} (S) :=     \initialres \left[ \e^{\i tH_\la } \,  ( S \otimes 1)  \,  \e^{-\i tH_\la } \right]
\eeq
whenever the RHS is well-defined.   
%

A mathematically precise definition of the reduced dynamics is the subject of the next lemma.

\begin{lemma}\label{lem: definition dynamics}
Suppose that Assumption \ref{ass: exponential decay} (see (\ref{eq: ass decay})) holds and define
\beq \label{def: evolved interaction}
H_{\sys\res}(t) :=    \sum_{x \in \lat}  1_x(t) \otimes \Psi_x(\e^{\i t \bosondispersion} \phi)  \quad \textrm{with} \quad 1_x(t):= \e^{\i t H_\sys}1_x   \e^{-\i t H_\sys} .
\eeq
The series\footnote{In fact, one needs to do things more carefully, since $ H_{\sys\res}(t) \notin \scrC$. A possible solution is to define the cut-off interaction $H_{\sys-\res, \La}(t) =    \sum_{x \in \La }   1_x(t)  \otimes \Psi_x(\e^{\i t \bosondispersion} \phi)$, for some finite subset $\La \subset \bbZ^d$, and to show that one can take the limit $\La \nearrow \bbZ^d$ in the expression analogous to \eqref{eq: dyson1}. }
    \beq \label{eq: dyson1}
 \caZ_t^{\la,*}(S):=   \mathop{\sum}_{n\in \bbZ^+  }     (\i \la)^{n}   
  \mathop{\int}\limits_{0 \leq t_1 \leq \ldots \leq t_n \leq t } \d t_1 \ldots \d t_n \, 
 \initialres \left[   \ad(H_{\sys\res}(t_1))  \ldots  \ad(H_{\sys\res} (t_n)) \,  \e^{\i t \ad(H_\sys)}     (S \otimes 1)  \,
  \right] 
    \eeq
 is well-defined for any $\la,t \in \bbR$ and arbitrary $S \in \scrB(\scrH_\sys)$, i.e., the RHS converges  absolutely in the norm of $\scrB(\scrH_\sys)$,  and $\caZ_t^{\la,*}$ has the expected properties, namely
 \beq
 \caZ_t^{\la,*}(1)=1, \qquad    \norm \caZ_t^{\la,*}(S) \norm \leq \norm S \norm.
 \eeq
\end{lemma}
One can prove this lemma (under  less restrictive conditions than those in Assumption \ref{ass: exponential decay})  by  direct estimates of the RHS of \eqref{eq: dyson1}. For this purpose, the estimates given in the present paper amply suffice. However, one can also define the system-reservoir dynamics as a dynamical system on a Von Neumann algebra through the Araki-Woods representation. This is the usual approach in the mathematical physics literature; see e.g.\ 
 \cite{jaksicpilletderezinski,jaksicpillet2,froehlichmerkli}.

Finally, we define $\caZ_{t}^{\la}: \scrB_1(\scrH_\sys) \to \scrB_1(\scrH_\sys)$, the reduced dynamics in  the Schr\"{o}dinger picture,  by  duality, i.e.,
\beq \label{eq: duality}
\Tr  [ \rho_\sys \caZ_{t}^{\la,*} (S) ]  =     \Tr  [ \caZ_{t}^{\la}(  \rho_\sys ) S   ], \qquad   S \in \scrB(\scrH_\sys), \rho_\sys \in \scrB_1(\scrH_\sys).
\eeq

We could also have started by defining the full initial state $\rho_{\sys\res}$ of the total system  consisting of the particle and reservoirs as  the positive, normalized functional 
\beq \label{def: product state}
\rho_{\sys\res} := \rho_\sys \otimes \rho_\res^\be  \qquad  \textrm{ on} \quad  \scrB(\scrH_\sys) \otimes \scrC,
\eeq
where we abuse notation by employing the same symbol $\rho_\sys$  for both  the density operator (a positive  element of $\scrB_1(\scrH_\sys)$) and  the state it determines on $\scrB(\scrH_\sys)$, i.e.,
 \beq  
 \rho_\sys [S]:= \Tr[ \rho_\sys S], \qquad    S \in  \scrB(\scrH_\sys).
 \eeq
Then,
\beq
\rho_{\sys\res} \left[ \e^{\i tH_\la } \,  ( S \otimes 1)  \,  \e^{-\i tH_\la } \right]=  \Tr  [ \caZ_{t}^{\la}(  \rho_\sys ) S   ].
\eeq
In what follows, we simply write $\rho$ for $\rho_\sys$.
For convenience, we treat $\rho$ as an element of the Hilbert space $\scrB_2(\scrH_\sys)$, which is justified since $\scrB_1(\scrH_\sys) \subset \scrB_2(\scrH_\sys)$.

\section{Result} \label{sec: result}

We now state our main results.  Recall that the position operator $X$ on $l^2(\bbZ^d)$ is given by
\beq
(Xf)(x)= x f(x), \qquad  x \in \lat, f \in l^2(\lat).
 \eeq
For $\ka  \in \bbC^d$, we define
\beq \label{def: jkappa}
\caJ_{\ka} S :=    \e^{ -\frac{\i }{2} (\ka, X)   } \,S\,    \e^{ -\frac{\i }{2} (\ka, X)   } , \qquad  S \in \scrB(\scrH_\sys).
\eeq
Note that $\caJ_{\ka}$ is unbounded if $\ka \notin \bbR^d$.
We choose an initial state $\rho \in \scrB_1(\scrH_\sys)$  satisfying 
\beq \label{eq: conditions initial state}
 \rho > 0, \qquad   \Tr [\rho]=1 \qquad    \norm \caJ_{\ka} \rho \norm_2 < \infty,
\eeq
for $\ka  $ in some open neighborhood of $0 \in \bbC^d$.

Our first result says that the momentum distribution of the particle tends to a stationary distribution exponentially fast.

\bet\label{thm: stationary}\emph{[Equipartition Theorem]}
Suppose that Assumption \ref{ass: analytic dispersion} (see Section \ref{sec: particle}) and Assumption \ref{ass: exponential decay} (see (\ref{eq: ass decay})) hold, and let $\rho$ satisfy condition \eqref{eq: conditions initial state}. There are positive constants $\la_0>0$ and $ g>0$  such that for $0 <\str\la\str \leq \la_0$,  there is a function $\zeta^0_\la \in L^2(\tor)$ satisfying
\beq
\Tr [ \theta(P)  \caZ_{t}^\la (\rho) ]    =     \langle  \theta, \zeta^0_\la  \rangle_{L^2(\bbT^d)} + O(\norm \theta \norm_2  \e^{-\la^2 g t}) , \qquad  \textrm{as} \,  t \nearrow \infty, \, \qquad  \textrm{for any} \,  \theta=\overline{\theta} \in L^{\infty}(\tor),
\eeq
and
\beq
 \zeta^0_\la(k) =    \frac{\e^{-\beta \systemdispersion(k) }}{  \int_{\bbT^d} \d k \,  \e^{-\be \systemdispersion(k)} }+ O(\la^{2}), \qquad   \la \searrow 0.
\eeq
\eet
The decay rate $\la^2 g$  is strictly smaller than $g_\res$, introduced in \eqref{eq: ass decay}.

 Define  a probability density  $\mu_t^\la$ depending on the initial state $\rho \in \scrB_1(\scrH_\sys)$ by
 \beq \label{def: density}
 \mu_t^\la(x) :=     \Tr \left[   1_x   \caZ^\la_t (\rho)   \right].
 \eeq
 It  is easy to see that
 \beq
 \mu_t^\la(x) \geq 0, \qquad    \sum_{x \in \lat} \mu_t^\la(x) =  \Tr [\rho]= 1 .\eeq

We claim that the particle exhibits a diffusive motion. This is the content of the next result.
 \begin{theorem}\label{thm: diffusion}\emph{[Diffusion]}
 Under  the same assumptions as in Theorem \ref{thm: stationary}, the following holds.  Let the initial state $\rho$ satisfy  condition \eqref{eq: conditions initial state} and let $\mu_t^\la$ be as defined in \eqref{def: density}. There is a  positive constant $\la_0$ such that, for $0 <\str\la\str \leq \la_0$,
 \beq \label{eq: diffusion}
      \sum_{x \in \lat}   \mu_t^\la(x)  \e^{- 
     \frac{ \i }{\sqrt{t}} (q, x) }  \,    \mathop{\longrightarrow}\limits_{t \nearrow \infty} \,   \e^{-\frac{1}{2}(q,D_\la q)}, \qquad  q \in \bbR^d
 \eeq
 where the \emph{diffusion matrix} $D_\la$ is  positive-definite (i.e., has strictly positive eigenvalues), and 
 \beq
 D_\la = \lakl \left( D_{\mathrm{kin}} + O(\la^2) \right),  \qquad  \la \searrow 0,
 \eeq
 with $D_{\mathrm{kin}}$ a $\la$-independent  positive-definite matrix  introduced in Section \ref{sec: kinetic limit}.
  \end{theorem}
We refer to Section \ref{sec: intro} for an explanation of  the connection between  this result and diffusion in the physicists'  sense. We close this section with some remarks concerning possible extensions of our results.
\begin{remark}\label{rem: convergence of moments}
Our proof of Theorem \ref{thm: diffusion} actually gives a stronger result. Assume the $n$th moments of the initial distribution  are bounded, or, equivalently,  
\beq \label{derivatives of characteristic function}
 q \mapsto  \sum_{x \in \lat}   \mu_0^\la(x)  \e^{- 
      \i (q, x) }    \quad  \textrm{is $n$ times differentiable.} 
\eeq
Then the rescaled $n$th moments  converge to the $n$th moments of the limiting distribution, or equivalently, the derivatives of $n$th order of  
\beq
 q \mapsto  \sum_{x \in \lat}   \mu_t^\la(x)  \e^{- \frac{\i}{\sqrt{t}}  (q, x) } 
\eeq
converge, as $t \nearrow \infty$, to the derivatives of $ \e^{-(q,D_\la q)}$.
For $n=2$, this implies \eqref{eq: convergence of second moment}.
Note that the condition \eqref{derivatives of characteristic function} is a weaker assumption than  \eqref{eq: conditions initial state}; in fact, \eqref{eq: conditions initial state} implies that \eqref{derivatives of characteristic function} is a real-analytic function.
\end{remark}

\begin{remark}
By the same technique as employed in our proofs, one can  show that correlations decay rapidly in time. As explained in the introduction, this rapid decay provides an intuitive explanation why the particle motion is diffusive.

Define the particle velocity operator by
\beq
V(t):= \i \e^{\i t H_\la} [H_\la, X]   \e^{-\i t H_\la} 
\eeq
and observe that 
\beq
V(0)= \i [H_\la, X] = \i  [H_\sys,X] =  (\nabla \systemdispersion)(P).
\eeq
Suppose that Assumptions \ref{ass: analytic dispersion} and \ref{ass: exponential decay} hold and let $\rho=\rho_\sys $ satisfy condition \eqref{eq: conditions initial state}. 

By  reasoning similar to that in Lemma \ref{lem: definition dynamics}, one can define the velocity-velocity correlation function $\rho_{\sys\res} \left[ V(t_1) V(t_2) \right] $.
Let the coupling strength $\la$ and the positive constant $g$ be as in Theorem \ref{thm: stationary}. Then, for all  $0 \leq t_1, t_2 < \infty$,
\beq\label{eq: decay of correlations}
 \big\str \rho_{\sys\res} \left[ V(t_1) V(t_2) \right]  \big\str \leq  c\,  \e^{- \la^2 g \str t_2-t_1 \str}, \qquad  \textrm{for some} \, c <\infty .
\eeq
\end{remark}

\begin{remark}
The condition that the particle dispersion satisfies  $ \systemdispersion(k)=\systemdispersion(-k)$ is not really necessary for our results to hold. If one did not impose this condition, the particle could have a drift velocity $v_{\mathrm{dr}}$ given by
\beq
v_{\mathrm{dr}} :=   \langle \nabla \systemdispersion, \zeta^0_\la  \rangle,
\eeq
and the particle motion  would still be diffusive, but one would now consider the "random variable" $ \frac{1  }{ \sqrt{t}}(x_t- v_{\mathrm{dr}}t)$, instead of  $ \frac{ 1 }{ \sqrt{t}}x_t$. In other words, in  \eqref{eq: diffusion}, one would have to replace 
\beq
   \sum_{x \in \lat}   \mu_t^\la(x)  \e^{- 
     \frac{ \i }{\sqrt{t}} (q, x)}    \qquad  \textrm{by}  \qquad      \sum_{x \in \lat}   \mu_t^\la(x)  \e^{- 
     \frac{ \i }{\sqrt{t}} (q, (x- v_{\mathrm{dr}} t )) } .
\eeq
Similarly, in eq.\  \eqref{eq: decay of correlations}, one would have to replace $
  V(t) $ by  $    V(t)- v_{\mathrm{dr}} $.
\end{remark}

\section{Discussion and outline of the proof} \label{sec: discussion}

\subsection{Translation invariance}\label{sec: fiber decomposition}

Consider the space of Hilbert-Schmidt operators $\scrB_2(\scrH_\sys)\sim \scrB_2(l^2(\lat))  \sim L^2(\bbT^d \times \bbT^d, \d k_1\d k_2)$, and define 
\beq
    \hat S     (k_1,k_2) :=\frac{1}{(2\pi)^{d}} \sum_{x_1,x_2 \in \lat}    S(x_1,x_2)  \e^{- \i (x_1, k_1)+\i (x_2 ,k_2) }      , \qquad  S \in     \scrB_2(l^2(\lat)).
\eeq
In what follows, we simply write $S$ for $\hat S$.
To deal conveniently with the translation invariance in our model, we make the change of variables 
\beq
k= \frac{k_1+k_2}{2}, \qquad  p=k_1-k_2,
\eeq
and,   for a.e.\ $p \in \tor$, we obtain a well-defined function $ S_p \in L^2(\bbT^d)$ by putting
\beq \label{def: Sfiber}
(S_p)(k) := S (k+\frac{p}{2},k-\frac{p}{2}).
\eeq
This follows from the fact that the Hilbert space  $\scrB_2(\scrH_\sys) \sim L^2(\bbT^d \times \bbT^d, \d k_1\d k_2)$ can be represented as a direct integral
\beq \label{def: fiber decomposition}
\scrB_2(\scrH_\sys) = \int_{\oplus \bbT^d} \d p \,    \scrH^p , \qquad     S =  \int_{\oplus \bbT^d} \d p \, S_p,
\eeq
where each `fiber space' $\scrH^p$ is naturally identified with $L^2(\bbT^d)$.  Let $\caT_z, z \in \lat$, be the lattice translation
\beq 
(\caT_z S)(x_1,x_2) : = S(x_1+z,x_2+z), \qquad  S \in \scrB(\scrH_\sys),
\eeq
or, equivalently,
\beq 
(\caT_z S)_p(k)  =  \e^{\i (p, z) }S_p, \qquad  S \in \scrB(\scrH_\sys).
\eeq
Since $H_\la$ and $\initialres$ are translation invariant, it follows that
\beq \label{translation inv of Z}
   \caT_{-z} \caZ_{t}^\la  \caT_z =  \caZ_{t}^\la.
\eeq

Let $\caW \in \scrB(\scrB_2(\scrH_\sys))$ be translation invariant in the sense of  eq.\ \eqref{translation inv of Z}, i.e., $  \caT_{-z} \caW \caT_z=\caW$.
Then it follows that, in the representation defined by \eqref{def: fiber decomposition}, $\caW$ acts diagonally in $p$, i.e.\ 
  $(\caW S)_p$ depends only on $ S_p $, and we define 
${\caW }_p$ by 
\beq
({\caW S})_p  = {\caW} _p   S_p.
\eeq
For the sake of clarity, we give  an explicit expression for ${\caW }_p$. 
Define the kernel $\caW(x,y;x',y')$ by 
\beq
  (\caW S)(x',y')=\sum_{x,y \in \lat}    \caW(x,y;x',y')   S(x,y), \qquad   x',y' \in \lat.
\eeq
Translation invariance is expressed by 
\beq
\caW(x,y;x',y') = \caW(x+z,y+z;x'+z,y'+z),   \qquad  z \in \lat,
\eeq
and, as an integral kernel,  $ {\caW}_p \in \scrB(L^2(\tor))$  is given  by
\beq
 {\caW}_p (k',k)=       \mathop{  \sum}\limits_{\left.\begin{array}{c}x,y, x',y' \in \lat  \\  x+y+x'+y'=0  \end{array}\right.} \e^{\i (k,x-y ) -\i (k' , x'-y') }    \e^{  \frac{\i}{2}(p, (x'+y')- (x+y))}   \caW(x,y;x',y').
\eeq

Next, we state an easy lemma.
\begin{lemma} \label{lemma l2-l1}
Let $S \in \scrB_1(\scrH_\sys)$. Then,  $S_p$, as defined in \eqref{def: Sfiber}, is well-defined  as a function in $L^1(\tor)$  for every $p$, and 
\beq \label{eq: trace as integral}
\Tr [  \caJ_{p} S   ] = \sum_{x \in \lat} \e^{-\i p x} S(x,x) = \langle 1, S_p \rangle.
\eeq
where $1 \in L^2(\bbT^d)  \cap L^{\infty}(\bbT^d)$ is the constant function with value $1(k)=1$.
Assume, moreover, that  there is a constant $\delta>0$ such that  
\beq
 \norm \caJ_{\ka} S \norm_2  < \infty\qquad \textrm{for} \qquad   \str \Im \ka \str <  \delta,
\eeq
then the function $p \mapsto S_p \in L^2(\tor)$ has a bounded-analytic extension to the strip $\str \Im p  \str < \delta$. 
\end{lemma}

The first statement of the lemma follows from the singular-value decomposition for trace-class operators and standard properties of the Fourier transform. In fact, the correct statement asserts that one can \emph{choose} $S_p$ such that \eqref{eq: trace as integral} holds. Indeed, one can change the value of the kernel  $S(k_{1},k_{2})$ on the line $k_{1}-k_{2}=p$ without changing the operator $S$, and hence $S_p$ in \eqref{eq: trace as integral} can not be defined  via \eqref{def: Sfiber} for all $p$, but only for almost all $p$.

 The second statement of Lemma \ref{lemma l2-l1} is the well-known relation between exponential decay of functions and analyticity of their Fourier transforms.   Since we will always demand the initial density matrix $\rho_0$ to be such that $\norm \caJ_\ka \rho_0\norm_2 $ is finite for $\ka$ in a complex domain, we will mainly need the second statement of Lemma \ref{lemma l2-l1}.

\subsection{Return to equilibrium inside the fibers} \label{sec: return to eq}

The main idea of our proof is that the reduced evolution in the `low momentum fibers', $ ({\caZ_{t}^\la})_p$, for $p $ near $0$,  has an invariant state to which every well-localized initial state relaxes exponentially fast.

Recalling that $H_\sys=\systemdispersion(P)$ and that the system is weakly coupled to a heat bath at inverse temperature $\beta$, we expect that, in an appropriate sense, and for arbitrary initial states $\rho \in \scrB_1(\caH_\sys)$,
\beq \label{return1}
\caZ^\la_t (\rho)  \, `` \mathop{\longrightarrow}\limits_{t \uparrow \infty} " \,   \frac{1}{Z(\be)}\e^{-\be \systemdispersion(P)}+ o(\la^0), \qquad  \la \searrow 0.
\eeq
We observe that $\e^{-\be \systemdispersion(P)} \notin \scrB_1(\caH_\sys)$, hence \eqref{return1} cannot  hold  in norm (in other words, $Z(\be)=\infty$).
One way to interpret \eqref{return1} is that it gives the correct asymptotic expectation value of functions of the momentum, and that is exactly what Theorem \ref{thm: stationary} states. 

  For every $\rho$ satisfying \eqref{eq: conditions initial state}, we have that
\beq\label{expectation as scalar product}
\Tr [\overline{\theta}(P) \caZ_{t}^\la (\rho)]=\langle \theta,  ({\caZ_{t}^\la \rho})_0 \rangle, \qquad \theta \in L^{\infty}(\tor),
\eeq
by applying  Lemma \ref{lemma l2-l1} with $S:=\overline{\theta}(P) \caZ_{t}^\la (\rho) $.
 Hence,  we should apparently attempt  to prove `return to equilibrium' for the evolution $({ \caZ_{t}^\la})_0$ on $L^2(\tor)$.

The  dynamics in the fibers corresponding to small values of $p$  provides information on the diffusive character of the system.  The probability density $\mu^\la_t(x)$ corresponding to some initial state $\rho$ is defined as in \eqref{def: density} . By Lemma \ref{lemma l2-l1}, 
\beq\label{eq: muS}
\sum_{x \in \bbZ^d}  \mu^\la_t(x) \e^{-\i (p,x) } =  \sum_{x \in \bbZ^d}    (\caZ_t^\la \rho)(x,x)  \e^{-\i (p,x) }   =   \int_{\bbT^d} \d k    (\caZ_t^\la \rho)  (k+\frac{p}{2}, k-\frac{p}{2}) = \langle 1, ({ \caZ_t^\la}\rho)_p  \rangle.
\eeq
To establish diffusion, it suffices to show that, for $\la$ fixed and for $p$ in a neighborhood of $0 \in \bbT^d$,
\beq \label{cltfromld}
\langle 1, ({ \caZ_t^\la}\rho)_p  \rangle   =     \e^{ t (-\frac{1}{2}(p, D_\la p )+ o(p^2) )  }  (1  + o(t^0)+o(p^0)), \qquad   t \nearrow \infty, p \searrow 0.
\eeq
for some  positive-definite matrix $D_\la$.
 Indeed,  by \eqref{eq: muS}, Theorem \ref{thm: diffusion} follows from \eqref{cltfromld} by  taking $p=\frac{q}{\sqrt{t}}$.  Thus, in order to prove Theorem \ref{thm: diffusion}, we are led to study the long-time asymptotics of the evolution $({ \caZ_t^\la})_p$, for small $p$. 
 
However, as our approach is perturbative in $\la$, expression \eqref{cltfromld} is not a good starting point, since  $(p, D_\la p )= O(\lakl)$, for fixed $p$ (as can be seen from the statement of Theorem \ref{thm: diffusion}), and hence one cannot perturb around $(p, D_\la p ) \big\str_{\la=0}$.   The way out of this difficulty is to set up the perturbation on a scale where the diffusion constant is finite (this will turn out to be the kinetic scale), or, in other words, to take the $p$-neighborhood in \eqref{cltfromld}  to shrink, as $\la \searrow 0$.   Since $\la$ approaches $0$, one must wait a time of order $\lakl$, before one sees the effect of the interaction.  Since, between collisions, the velocity of the free particle is unaffected,  it travels a distance of order $\lakl$ . This means that when both space and time are  measured in units of $\lakl$;
\beq
x= \lakl \tilde x_\la,    \qquad  t= \lakl  \tilde t_\la,
\eeq
we expect a diffusion constant  $\tilde D_\la \sim \frac{ (\tilde x_\la)^2}{\tilde t_\la}  $ of order $O(1)$. This is consistent with the fact that $D_\la \sim \frac{ x^2}{t}  $ is of order $\lakl$. The limit
$ \tilde D_{\la \searrow 0}$ is the diffusion constant in the kinetic limit, as outlined in the next section.

\subsection{The kinetic limit}  \label{sec: kinetic limit}
To control the asymptotics of  the effective time-evolution $ ({\caZ_{t}^\la})_p $, we compare it with the corresponding evolution in the \emph{kinetic limit}, which is the limit approached when microscopic space and time are taken to be  
$ \lakl x, \lakl t $, respectively, and  the coupling strength $\la \rightarrow 0$;  as announced in the previous section.  It has been proven in \cite{erdos} (for models with only one thermal reservoir) that,  in this limit,  the dynamics is described by a linear Boltzmann equation. 

Our variant of this result is described below.
\subsubsection{Convergence to a linear Boltzmann equation}
The effective reservoir structure factor $\psi$ has been defined in (\ref{def: correlation function}-\ref{def: effective structure factor}).  For convenience, we  introduce a positive function $r(\cdot,\cdot)$, with
\beq
r(k,k'):=  \psi  [\systemdispersion(k') -\systemdispersion(k) ]       \geq 0.
\eeq
  For $\ka \in \bbR^d$, we define a bounded linear operator, $M^\ka$, on $L^2(\bbT^d)$  by
\beq  \label{def: operator M}
(M^\ka \theta)(k)  := \i  (\ka,\nabla \systemdispersion)(k) \theta (k)  + \int_{\tor} \d k'   \, \left[  r(k',k)    \theta(k')  -     r(k,k')  \theta (k )     \right] , \qquad   \theta \in L^2(\bbT^d),
\eeq
where $ (\ka,\nabla \systemdispersion)(k)$ stands for the scalar product  in $\bbC^d$ of  $\ka $ and $\nabla \systemdispersion(k)$.
The operator $M^\ka$ has a straightforward interpretation:  Consider a classical particle whose states are specified by a position $x \in \bbR^d$ and a 'momentum' $k \in \bbT^d$.
The momentum $k$ evolves according to a Poisson process with a rate $ r(k,k')$ for the transition from state $k$ to $k'$. Between two momentum jumps, the particle moves freely,  with speed given by $(\nabla \systemdispersion)(k)$. 
  The translation of this picture into a mathematical statement is as follows:
The state-space distribution of  the classical particle at time $t$ is given by a probability density $\nu_t(\cdot,\cdot)$ on $\bbR^d \times \tor$; ($\nu(x,k) \geq 0$ and $\int \d x\d k \,  \nu_t(x,k)=1$).  Then
\beq \label{eq: evolution eq nu} 
 \frac{\partial}{\partial t}   \nu_t(x,k)    =  ( \nabla_k  \systemdispersion , \nabla_x \nu_t) (x,k)    +  \int_{\tor} \d k'   \, \left[  r(k',k)    \nu_t(x,k')  -     r(k,k')  \nu_t(x,k )     \right].
  \eeq 
One checks that
  \beq
  \hat\nu_t^\ka(k):= (2\pi)^{-d/2} \int_{\bbR^d} \d x \,   \e^{-\i (\ka, x)}     \nu_t(x,k) 
  \eeq
satisfies an evolution equation generated by $M^\ka$; \beq \label{eq: evolution eq nu2}   \frac{\partial}{\partial t}  \hat\nu_t^\ka= M^\ka \hat\nu_t^\ka.   \eeq 
We claim that the rates $r(k,k')$ satisfy the identity
\beq
r(k,k') = r(k',k)\e^{-\be( \systemdispersion(k')-\systemdispersion(k))},
\eeq
known as  the \emph{detailed balance condition} in the context of Markov processes.  
It is a direct consequence of the KMS-condition for the reservoirs. In our context, it is easily derived from  \eqref{def: thermal state canonical}.
  The detailed balance condition implies that 
 \beq \label{def: steady state}
 M^0 \zeta_{\mathrm{kin}}^{0} =0,      \qquad  \textrm{where} \quad  \zeta_{\mathrm{kin}}^{0} (k)    =     \frac{ \e^{- \be \systemdispersion(k)} }{\int_{\tor} \d k \e^{-\be \systemdispersion(k)}}.
\eeq
 In the language of Markov processes,  $ \zeta_{\mathrm{kin}}^{0} $ is a stationary state.

The relevance of $M^\ka$ is that it describes the evolution $\caZ_{\lakl t}^\la$ in the fiber indexed by $\la^2 \ka$ in the limit $\la \searrow 0$. Moreover, the convergence of the fiber dynamics $(\caZ_{\lakl t}^\la)_{\la^2 \ka}$ holds even after analytic continuation to complex $\ka$. One can prove the following result
\begin{proposition}\label{thm: wc limit small fibers}
Assume  Assumptions \ref{ass: analytic dispersion}  and \ref{ass: exponential decay}. Then, for $\str\Im\ka\str$ sufficiently small and $    0<t<\infty $,
\beq    \left\norm     ({ \caZ_{ \lakl t}^{\la} })_{\la^2\ka} -  \e^{tM^{\ka} } \right\norm \,\mathop{\longrightarrow}\limits_{\la \searrow 0} \, 0.    \eeq
where the norm is the operator norm on $L^2(\tor)$.
\end{proposition}
We do not prove this Proposition (which is not needed for the proof of our results). In fact, the proof is based on  the same reasoning as in Section \ref{sec: proof of main technical}.
Of course, one can also express Proposition \ref{thm: wc limit small fibers} in terms  of the rescaled Wigner function, as is done in  \cite{erdos,erdosyauboltzmann}.   
Indeed, setting
\beq \label{convergence wigner}
\hat\al_t^\ka(k) := \lim_{\la \searrow 0}  \left(\caZ^\la_{\lakl t}\rho\right)( k+\la^2 \frac{\ka}{2}, k-\la^2 \frac{\ka}{2}  )=    \lim_{\la \searrow 0}  ({\caZ^\la_{\lakl t}\rho})_{\la^2\ka}(k),
\eeq
one obtains from Proposition \ref{thm: wc limit small fibers} that $\hat\al_t^\ka(k)$ satisfies the evolution equation   \eqref{eq: evolution eq nu2}. (It would thus be justified to call $\hat\al_t^\ka(k)$ simply $\hat\nu_t^\ka(k)$). Its inverse Fourier transform 
\beq
\al_t(x,k) =  (2\pi)^{-d/2}  \int_{\bbR^d} \d \ka \,   \e^{\i (\ka, x) }  \hat\al_t^\ka(k)
\eeq
 is a probability density on $\bbR^d \times \tor$ and satisfies  \eqref{eq: evolution eq nu} with initial condition $\al_0(x,k) = \delta(x) \rho(k,k)$.
 
We state another useful consequence of Proposition \ref{thm: wc limit small fibers}. Recall that the probability density $\mu_t(\cdot)$ has been defined in \eqref{def: density}, for any initial state $\rho$.  
Taking the scalar product with $1 \in L^2(\bbT^d)$ on both  sides of \eqref{convergence wigner} and using \eqref{eq: muS}, we obtain that
\beq \label{eq: kinetic limit for mu}
   \sum_{x \in \lat}   \e^{-\i \la^2 (\ka, x) }\mu_{\lakl t}^\la(x)     \,\mathop{\longrightarrow}\limits_{\la \searrow 0}    \mathop{\int}\limits_{\tor} \d k \,  \hat\al^\ka_t(k).
\eeq
As outlined in Section \ref{sec: return to eq}, the $t \nearrow \infty$ asymptotics of  the LHS of \eqref{eq: kinetic limit for mu} contains information on the diffusive behavior of the particle. In the next section  we discuss the $t \nearrow \infty$ asymptotics of the RHS of \eqref{eq: kinetic limit for mu}.

\subsubsection{Diffusive behavior of solutions of the Boltzmann equation}

To realize that the Boltzmann equation describes diffusion, one studies the spectral properties of $M^\ka$, for small $\ka$. 
We state a crucial result, Theorem \ref{thm: main technical wc},  and we refer the reader  to \cite{clarkderoeckmaes} for complete proofs and a more extended  discussion of quantum dissipative evolutions.
\bet \label{thm: main technical wc}
Suppose that Assumptions \ref{ass: analytic dispersion} and \ref{ass: exponential decay} hold, and let $M^\ka \in \scrB(L^2(\tor))$ be defined as in \eqref{def: operator M}.  

Then there is a  positive constant $\delta_{\mathrm{kin}}$ such that the operator $ M^{ \ka} $, with $\str \ka \str \leq \delta_{\mathrm{kin}}$,   has a simple eigenvalue, $  f_{\mathrm{kin}} ( \ka) $, separated from the rest of the spectrum by a gap, 
\beq
\mathrm{dist}( \Re  f_{\mathrm{kin}} ( \ka) ,  \Re \Omega)  =:  g_{\mathrm{kin}}>0     \label{def: omega and gap}
\eeq
where 
\beq
  \Om :=   \cup_{\str\ka\str <  \delta_{\mathrm{kin}}  } \left(\sp M^\ka \setminus \{ f_{\mathrm{kin}}(\ka)\} \right)   ,\qquad  \textrm{and} \, \Re \Om <0. 
\eeq
  The eigenvalue $ f_{\mathrm{kin}} $ and its associated  eigenvector $ \zeta_{\mathrm{kin}}^{\ka}$  and spectral projection $P_{\mathrm{kin}}^{\ka}$ are analytic in $\ka$  and
 \beq \label{Dkin as hessian}
 f_{\mathrm{kin}}( \ka) =-  \big\langle 1,  (\ka,\nabla\systemdispersion) \big(M^0\big)^{-1}(\ka, \nabla\systemdispersion) \zeta_{\mathrm{kin}}^{0} \big\rangle  +O(\ka^3), \qquad   \ka \searrow 0,
\eeq
where  $\nabla\systemdispersion$  denotes the operator that acts by multiplication with the function $\nabla\systemdispersion$ on $\tor$.   
The diffusion matrix, $D_{\mathrm{kin}}$,   defined by
\beq \label{def: Dkin}
(D_{\mathrm{kin}})^{i,j} :=  -     \frac{\partial^2}{\partial \ka^i \partial \ka^j} f_{\mathrm{kin}}(\ka)\big\str_{\ka=0}, \qquad i,j =1,\ldots d,
\eeq
has real entries  and is  positive-definite.
\eet

\noindent\emph{Sketch of proof.} \newline
We write  $M^0=K+T$  with
\beq
(K\theta)(k) = \int_{\tor} \d k' r(k',k) \theta(k') , \qquad
(T\theta)(k) = -\left(\int_{\tor} \d k' r(k,k')\right) \theta(k), \qquad  \theta \in L^2(\tor,\d k).
\eeq
Notice that $T$ is a multiplication  operator with spectrum 
\beq
\sp T = \left\{ - \int_{\tor} \d k' r(k,k') \, \big\str  \,  k \in \tor  \right\}, \qquad    r(k,k')\geq 0.
\eeq
The operators $K$ and $T$ are sometimes  referred to as the \emph{gain} and \emph{loss} terms in the Boltzmann  equation.
Assumptions  \ref{ass: analytic dispersion} and \ref{ass: exponential decay} imply that the functions $\psi$ and $\systemdispersion$  are  real-analytic in $k$,  and hence  $r(\cdot,\cdot)$ is  real-analytic in both variables.  It follows that $K$ is a compact operator on $L^2(\tor)$  and, since we assumed that $\psi \not\equiv 0$,  we have that
$
 \sup \Re \sp T < 0
$.
  By  Weyl's theorem  on the stability of the essential spectrum (see e.g.\ p.\ 101 of \cite{reedsimon4}), we deduce that the spectrum of $M^0$ in the region $\Re z > \sup \Re \sp T$ consists of isolated eigenvalues of finite multiplicity. 
  From the pointwise positivity of $r(\cdot,\cdot)$, the Perron-Frobenius theorem and from the fact that $M^0$ generates a contractive semigroup on $L^1(\tor)$ we then conclude that the eigenvalue $0$ of $M^0$ is simple and that it is separated by a gap from  the rest of the spectrum.   
  The spectral projection $P_{\mathrm{kin}}^0$ is explicitly given by
\beq \label{explicit projection}
P^0_{\mathrm{kin}} \theta= \langle 1, \theta \rangle   \zeta_{\mathrm{kin}}^{0}, \qquad    \theta \in   L^2(\tor)
\eeq
with $\zeta_{\mathrm{kin}}^{0}$ as in \eqref{def: steady state}. 
The analyticity of $f_{\mathrm{kin}}(\ka)$ and $\zeta_{\mathrm{kin}}^{\ka}$ is proven with the help of analytic perturbation theory.  
Using the assumption that $\systemdispersion(k)=\systemdispersion(-k)$, we check that
\beq \label{vanishing first order}
P^0_{\mathrm{kin}}   \nabla\systemdispersion P^0_{\mathrm{kin}}=0.
\eeq
Employing explicit expressions of second order perturbation theory, we obtain  formula \eqref{Dkin as hessian} as a consequence of the fact that $M^\ka-M^0=\i (\ka,\nabla\systemdispersion) $ and \eqref{vanishing first order}. 

Since $\overline{f_{\mathrm{kin}}(\ka)}=f_{\mathrm{kin}}(-\overline{\ka})$, it follows that the matrix $D_{\mathrm{kin}}$ has real entries.   The positive-definiteness of $D_{\mathrm{kin}}$ is established as follows. 
Consider the bounded operator 
\beq
(W \theta)(k)= \e^{\frac{1}{2}\be \systemdispersion(k)}\theta(k), \qquad   \theta \in L^2(\tor),
\eeq
and notice that $ \tilde M:=
W^{-1}M^0 W $ is a self-adjoint operator on $L^2(\tor)$, in particular $ \tilde \zeta:= W \zeta_{\mathrm{kin}}^0= W^{-1} 1$, (i.e.,  the left and right eigenvector corresponding to the eigenvalue $0$ are identical). For $\ka \in \bbR^d$, we can rewrite  \eqref{def: Dkin} as
\beq \label{eq: sa rep of diff const}
( \ka,D_{\mathrm{kin}}\ka) = \big\langle  (\ka,\nabla \systemdispersion)  \tilde \zeta,    \big(\tilde M\big)^{-1}       (\ka,\nabla \systemdispersion) \tilde \zeta \big\rangle.
\eeq
By Assumption \ref{ass: analytic dispersion},  the function $k \mapsto (\ka, \nabla \systemdispersion(k))$ does not vanish identically on $\tor$ (for $\ka \neq 0$).  Hence, by the spectral theorem applied to $\tilde M$,  expression \eqref{eq: sa rep of diff const} is strictly positive.
\qed
\vspace{0.5cm}

Let $\hat\nu_t^{\ka}(k)$ be a solution of the evolution equation  \eqref{eq: evolution eq nu2} for $\ka$ in some neighborhood of $0$ in $\bbC^d$. Using Theorem \ref{thm: main technical wc} and  reasoning similar to that in Section \ref{sec: return to eq}, it follows that    
\beq \label{eq: diffusion for nu}
 \int_{\tor} \d k \,  \hat\nu_t^{\ka}(k)  \quad  \mathop{\longrightarrow}\limits_{\ka=\frac{ q}{\sqrt{t}},\, t \nearrow \infty}^{}  \quad  \e^{-\frac{1}{2} (q, D_{\mathrm{kin}} q)}, \qquad  q \in \bbR^d.
\eeq
Hence  a solution $\nu_t(x,k)$ of the Boltzmann equation \eqref{eq: evolution eq nu}  behaves diffusively, with diffusion tensor $D_{\mathrm{kin}}$.

\subsection{Perturbation around the kinetic limit}\label{sec: perturbation around kinetic}

Up to now, we have seen that, in  the kinetic limit, the particle motion is described by a linear Boltzmann equation. Since solutions of the linear Boltzmann equation behave diffusively for large times (as is essentially stated in  Theorem \ref{thm: main technical wc}), we can associate a diffusion constant to our model. Indeed, by \eqref{eq: kinetic limit for mu}  and \eqref{eq: diffusion for nu},
\beq \label{eq: first la then t}
\lim_{t \nearrow \infty} \lim_{\la \searrow 0}\sum_{x \in \lat} \mu_{\lakl t}^\la (x) \e^{-\i \frac{ \la^2}{\sqrt{  t}} (q,x)   }         =  \e^{- \frac{1}{2}(q, D_{\mathrm{kin}}q)}     
\eeq
However, \eqref{eq: first la then t} does not give information on the long-time asymptotics of our system for small, but fixed $\str \la\str >0$.   The least one would wish for is to be able to exchange the order of limits in \eqref{eq: first la then t},  and, indeed, Theorem \ref{thm: diffusion} states that one can do so without affecting the RHS.  We stress this point, because it is an improvement of our paper when compared to most earlier results on diffusion.

Since we have learned that  $( { \caZ_{\lakl t}^\la})_{\la^2 \ka} $ has a well-defined limit, $\e^{t M^\ka}$,  as $\la \searrow 0$, (see Proposition \ref{thm: wc limit small fibers}), it is natural to expand $ ( { \caZ_{\lakl t}^\la})_{\la^2 \ka} $ around this limit, in such a way   that we can take  $t \nearrow \infty$.    We perform the expansion on the Laplace transform of $ \caZ^\la_t$,
\beq
 \caR_\la(z) := \int_{\bbR^+} \, \d t \, \e^{-tz}  \caZ^\la_t .
\eeq
Theorem \ref{thm: polymer model} below  summarizes the result of our expansion. Loosely speaking,  a key  consequence of this theorem is the fact  that, in the fibers indexed by $\la^2 \ka$,  one has that 
\beq
 ( \caR_\la(z))_{\la^2 \ka}= (z-\la^2 M^{\ka}- A(z,\la,\ka))^{-1},
\eeq
where  the operator $ A(z,\la,\ka)    $ is  ``small" compared to  $\la^2 M^\ka$.

\bet \label{thm: polymer model}
Suppose that Assumptions \ref{ass: analytic dispersion} and \ref{ass: exponential decay} in Section \ref{sec: model} hold. Then,
   there are operators  $\caL(z)$ and $ \caR^{\mathrm{ex}}_{\la}(z)$ in $\scrB(\scrB_2(\scrH_\sys))$ such that the following statements hold.
\ben
\item{ For  $(z,\la) \in \bbC\times \bbR$ satisfying $\Re z > \norm  \la^2\caL(z) + \caR^{\mathrm{ex}}_\la(z) \norm$, 
\beq \label{eq: expansion for resolvent}
 \caR_\la(z) = ( z-  \ad(\i H_\sys) - \la^2\caL(z) - \caR^{\mathrm{ex}}_\la(z)  )^{-1}.
\eeq 
}
\item{  The operators  $\caL(z)$ and $ \caR^{\mathrm{ex}}_{\la}(z)$ have the following properties: There are positive constants $\delta'_1, \delta'_2 , g' >0$ such that 
\beq
\caJ_{\ka}\caL(z)\caJ_{-\ka} , \qquad    \caJ_{\ka}\caR_\la^{\mathrm{ex}}(z) \caJ_{-\ka}
\eeq
are analytic in the variables $(z,\ka) \in \bbC\times\bbC^d$ in the region defined by   $\str \ka \str \leq \delta'_1, \Re z > -g', \str\la\str \leq \delta'_2$. Moreover, 
\baq
 & \mathop{\sup}\limits_{ \str \ka \str \leq \delta'_1, \Re z >- g'}  \norm  \caJ_{\ka}\caL(z)\caJ_{-\ka} \norm    = O(1),  &  \qquad \la\searrow 0  \\[2mm]
 &     \mathop{\sup}\limits_{ \str \ka \str \leq \delta'_1, \Re z >- g'} \norm    \caJ_{\ka}\caR_\la^{\mathrm{ex}}(z) \caJ_{-\ka}  \norm    =O(\la^4), & \qquad \la \searrow 0   \label{eq: bound rex },
\eaq
where $\norm \cdot \norm$ refers to the operator  norm  on  $\scrB(\scrB_2(\scrH_\sys))$ (as in (\ref{def: norm on operators})).
}

\item{ Let  $M^\ka$ be  defined as in Section \ref{sec: kinetic limit}. Then
\beq
\norm \left( \ad(\i  H_\sys) + \la^2 \caL(0)   \right)_{\la^2\ka} -  \la^2 M^{\ka} \norm = O(\la^4 \ka^2)+O(\la^4\ka ),  \qquad    \la^2\ka \searrow 0  , \la   \searrow 0.
\eeq
}
\een

\eet
The proof of Theorem  \ref{thm: polymer model}  is the subject of Section \ref{sec: dyson and proof of polymer}. From that proof, it becomes clear that  $g'$ can be chosen to be any fraction of $g_{\res}$ by
 making $\delta'_1$ and $\delta'_2$ small enough.

From Theorem \ref{thm: polymer model}, one obtains our main result by using Theorem \ref{thm: main technical wc} and standard analytic perturbation theory.   
More precisely, we prove the following theorem.

\bet \label{thm: main technical improved}
Suppose that Assumptions \ref{ass: analytic dispersion} and \ref{ass: exponential decay} in Section \ref{sec: model} hold. Then,
there are positive constants $\delta_1,\delta_2,g>0$ such that, for $(\la , \ka) \in \bbR\times \bbC^d$ and $\str \ka \str \leq  \delta_{1}, 0< \str \la \str \leq  \delta_{2}$, there is a rank $1$ operator  $P^{\la,\ka}$ and  a function $f(\la,\ka)$  satisfying
\beq \label{eq: return to equilibrium off fibers}
 \norm  ({\caZ_{t }^\la })_{\la^2 \ka} -  \e^{t f(\la,\ka )} P^{\la,\ka}   \norm   = O(\e^{t (f(\la,\ka)-\la^2 g)}), \qquad  t \nearrow \infty
\eeq
and
\beq
\norm P^{\la,\ka}-P^\ka_{\mathrm{kin}} \norm  = O(\la^2) ,    \qquad      \str f(\la,\ka)-  \la^2 f_{\mathrm{kin}}(\ka)   \str=  O(\la^4)   , \qquad  \la \searrow 0                                                        
\eeq  
Moreover,  $P^{\la,\ka}$ and $f(\la,\ka)$ are analytic in $\ka \in \bbC^d$ in the region defined by $\str \ka \str \leq  \delta_{1}, \str \la \str \leq  \delta_{2}$.
\eet 
By making $\delta_2$ small enough, the constant $g$ can be chosen to be any fraction of $g_{\mathrm{kin}}$ and $\delta_1$ can be chosen to be given by $\delta_{\mathrm{kin}}$, with $g_{\mathrm{kin}}, \delta_{\mathrm{kin}}$  as in Theorem \ref{thm: main technical wc}).

Theorems \ref{thm: stationary} (Equipartition Theorem) and \ref{thm: diffusion}  (Diffusion )  then follow as discussed in Section \ref{sec: return to eq}. We briefly recapitulate our reasoning.\\
\\
\emph{Proof of Theorems \ref{thm: stationary} and \ref{thm: diffusion}. } \\
We first prove Theorem \ref{thm: stationary}. Using  \eqref{expectation as scalar product} and Theorem \ref{thm: main technical improved}, we write, for $\theta=\overline{\theta} \in L^{\infty}(\tor)$,
\beq
 \Tr [ \theta(P) \caZ_{t }^\la \rho ] = \langle  \theta ,   ({\caZ_{t }^\la } \rho)_{0}    \rangle= \langle   \theta, \e^{t f(\la,0)} P^{\la, 0}       \rho_{ 0}  \rangle+ O(\e^{t (f(\la,0)-\la^2 g )}).
\eeq
Since $\caZ_{t }^\la \rho$  has trace $1$ (it is a density matrix)  for all $t \geq 0$, we deduce that $f(\la,0)=0$ and, setting $\theta=1$,
\beq \label{eq: trace preservation}
\langle  1 ,   P^{\la, 0}       \rho_{ 0}  \rangle=1.
\eeq
The fact that $ P^{\la, 0} $ is a rank $1$ operator (by Theorem \ref{thm: main technical improved}) implies, together with \eqref{eq: trace preservation},  that,
\beq
P^{\la, 0}  \eta=  \zeta_{\la}^0 \langle 1, \eta \rangle, \qquad  \textrm{for any}\, \eta \in L^2(\tor),
\eeq
 for some $\zeta_{\la}^0 \in L^2(\tor)$ which satisfies $\langle 1, \zeta_{\la}^0 \rangle=1$. 
Theorem \ref{thm: stationary} follows. \\

We define the diffusion matrix by
\beq
(D_\la)^{i,j} := - \la^{-4}   \frac{\partial^2}{\partial \ka^i \partial \ka^j} f(\la,\ka)\big\str_{\ka=0}, \qquad i,j =1,\ldots d.
\eeq
From \eqref{eq: trace as integral}, with $S:= \caZ_t^\la \rho$, we  conclude that   $\overline{f(\la,\ka)}=f(\la,-\overline{\ka}) $, and hence the matrix $D_\la$ has real entries. Positive-definiteness of $D_\la$ follows then from positive-definiteness of $D_{\mathrm{kin}}$, for $\la$ small enough.  
Using Theorem \ref{thm: main technical improved}, we  find that
 \baq
\sum_{x \in \lat} \e^{-\frac{\i  }{\sqrt{t}}  (q,x)   } \mu_{t}^\la(x)   &=&   \langle  1 ,   ({  \caZ_{ t}^\la \rho }  )_{ \frac{q}{\sqrt{  t}} }         \rangle \\
           &=&      \langle  1 ,   (  \caZ_{ t}^\la \rho )_{ \la^2 \ka}          \rangle, \qquad\qquad \qquad \textrm{with}\,  \ka= \la^{-2} \frac{q}{\sqrt{ t}}, \,   q \in \bbR^d    \nonumber \\
                      &=&        \langle  1 ,          \e^{ t f(\la, \ka   ) }   P^{\la,  \ka}      \rho_{ \la^2 \ka}          \rangle  (1+  O( \e^{ -g t })) , \qquad \qquad \textrm{as}\,   t \nearrow \infty  \nonumber\\[3mm]
                         &=&         \langle  1 ,          \e^{- t  ( \la^4 \frac{1}{2} (\ka, D_\la  \ka)  +O(\ka^3))  }   P^{\la, 0}        \rho_{ 0}   \rangle (1+ O(\ka))   (1+  O( \e^{ -g t })), \qquad \qquad \textrm{as}\, \ka \searrow 0  \nonumber\\[3mm]
                                                        &=&           \langle  1 ,          \e^{ -  \frac{1}{2}(q, D_\la q)  +O( t \ka^3)  }    P^{\la, 0}      \rho_{ 0}          \rangle     (1+ O(\ka))    (1+  O( \e^{ -g t })),         \nonumber     
  \eaq                              
  which proves Theorem \ref{thm: diffusion} upon 
  using  $\langle  1,  P^{\la, 0}       \rho_{ 0}  \rangle=1$  and    $  \ka= \la^{-2} \frac{q}{\sqrt{ t}} $. \qed

 Remark \ref{rem: convergence of moments} follows by standard reasoning, using the following facts
 \ben
 \item{  The family of operators \beq
 ({\caZ_{t }^\la })_{\la^2 \ka} -  \e^{t f(\la,\ka )} P^{\la,\ka}
  \eeq
  is analytic in $\ka$ in a neighborhood of $0 \in \bbC^d$ and bounded by a constant independent of $\ka$ and $t$. 
  }
 \item{ The function  $f(\la,\ka )$ and the rank $1$ operator  $P^{\la,\ka}$ are analytic in $\ka$ in a neighborhood of $0 \in \bbC^d$}
 \een
This  is related  to the general fact that the central limit theorem follows from the existence and analyticity of the large deviation generating function, as described in \cite{bryc}. Indeed,   $ \ka \mapsto f(\la,\ka )$ can  be viewed as the large deviation generating function corresponding to the family of random variables $x_t, t>0$, as defined in \eqref{def: random variable}.

\section{ Dyson expansion and proof of  Theorem \ref{thm: polymer model}}     \label{sec: dyson and proof of polymer}
To construct a "polymer model", we first write a Dyson expansion for $\caZ_{t}^\la$.

\subsection{ Dyson Expansion}\label{sec:  dyson expansion}

In this section, we set up a convenient notation to handle the Dyson expansion, which has been introduced in Lemma  \ref{lem: definition dynamics}.  
Define the unitary group $\caU_{t}  $   on  $\scrB_2(\scrH_\sys)$ by
\beq
\caU_{t}S := \e^{-\i t H_\sys} S \e^{\i t H_\sys} , \qquad  S \in \scrB_2(\scrH_\sys),
\eeq
and the operators $\caI_{x,l}$, with $x \in \lat$ and $l \in \{\links, \rechts \}$ ($\links,\rechts$ stand for ``\emph{left}" and ``\emph{right}"), as
\beq \label{def: operatorI}
\caI_{x,l}S:= \left\{ \begin{array}{rll}  \i & 1_x   S & \qquad \textrm{if} \quad l= \links \\[1mm]  -\i&
  S   1_x & \qquad \textrm{if} \quad l = \rechts. \\   \end{array} \right.
\eeq

Let $\caP_n$ be the set of partitions $\pi$ of the integers $1,\ldots,2n$ into $n$ pairs. We  write $(r,s) \in \pi$ if $(r,s)$ is one of these pairs, with the convention that $r<s$.   Note that the same notation was already used in \eqref{guassianity1} and in \eqref{eq: gaussian property1}.
Elements in $\bbR^{2n}, (\bbZ^d)^{2n}, \{ \links, \rechts \}^{2n}$ are denoted by $\underline{t}, \underline{x},\underline{l}$, with $t_i,x_i,l_i$  their respective components, for $ i=1,\ldots,2n$.  We evaluate \eqref{eq: dyson1} by using \eqref{def: thermal state canonical} and (\ref{eq: gaussian property1})-(\ref{eq: gaussian property2}):
\beq\label{eq: Z with pairings}
 \caZ^{\la}_{t} =    \sum_{n \in \bbZ^+} \la^{2n} \mathop{\int}_{0 \leq t_1 \ldots \leq t_{2n} \leq t}\, (\prod_{i=1}^{2 n}\d t_i)  \,  \mathop{\sum}\limits_{\underline{x} , \underline{l} }   \mathop{\sum}\limits_{\pi \in \caP_n}   \zeta_\pi(\underline{t},\underline{x}, \underline{l} )     \,      \caU_{t-t_{2n}}     \caI_{x_{2n},l_{2n} }    \ldots \caI_{x_2,l_2}   \caU_{t_2-t_1} \caI_{x_1,l_1}  \caU_{t_1}
\eeq
where 
\beq \label{def: zeta}
 \zeta_\pi(\underline{t},\underline{x}, \underline{l} )  := \prod_{(r,s) \in \pi}  \delta_{x_r,x_s}   \left\{  \begin{array}{cc}   \hat\psi (t_s-t_r )   & l_{r}=\links, \\ [1mm]  
 \hat\psi (-(t_s-t_r ))    & l_{r}=\rechts,
  \end{array}\right.
\eeq
and, for $n=0$, the integral in \eqref{eq: Z with pairings} is meant to be equal to $\caU_t$.

\begin{figure}[h!] 
\vspace{0.5cm}
\begin{center}
\psfrag{label0}{$t$}
\psfrag{labelt}{$0$}
\psfrag{label12}{  $\scriptstyle{(t_1,x_1, l_1)}$}
\psfrag{label11}{ $\scriptstyle{(t_2,x_2,l_2)}$}
\psfrag{label10}{ $\scriptstyle{(t_3,x_3, l_3)}$}
\psfrag{label9}{ $\scriptstyle{(t_4,x_4, l_4) }$}
\psfrag{label8}{ $\scriptstyle{(t_5,x_5, l_5) }$}
\psfrag{label7}{ $\scriptstyle{(t_6,x_6, l_6) }$}
\psfrag{label6}{ $\scriptstyle{(t_7,x_7, l_7) }$}
\psfrag{label5}{ $\scriptstyle{(t_8,x_8, l_8) }$}
\psfrag{label4}{ $\scriptstyle{(t_9,x_9, l_9) }$}
\psfrag{label3} { $\scriptstyle{(t_{10}, x_{10},, l_{10} )}$ }
\psfrag{label2} { $\scriptstyle{(t_{11}, x_{11}, l_{11} )}$ }
\psfrag{label1} { $\scriptstyle{(t_{12}, x_{12}, l_{12} )}$ }
\includegraphics[width = 18cm, height=6cm]{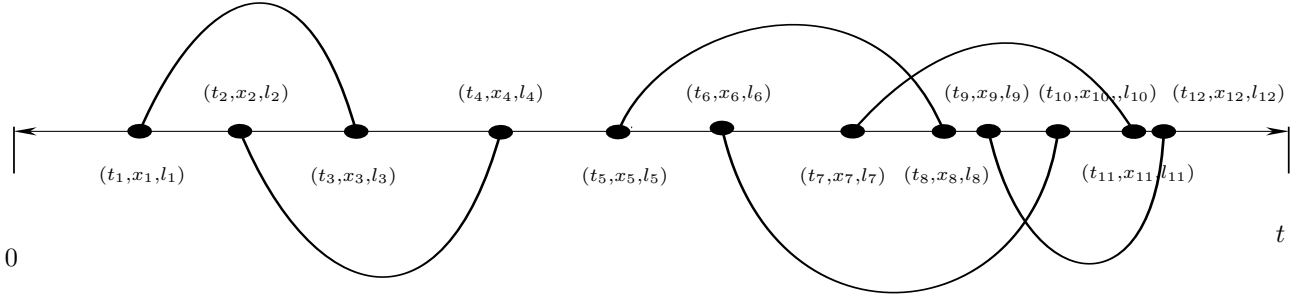}
\end{center}
\caption{ \footnotesize{Graphical representation of a  term  contributing to the RHS of \eqref{eq: Z with pairings} with $\pi=\{ (1,3), (2,4), (5,8 ), (6,10), (7,11), (9,12) \} \in \caP_6$. The times $t_i$ correspond to the position of the points on the horizontal axis.} \label{fig: explanation V} }
   \vspace{2mm}
     \begin{minipage}{18cm}
   \footnotesize{
 Starting from this graphical representation, we  can reconstruct the corresponding term in \eqref{eq: Z with pairings} - an operator on $\scrB_2(\scrH_\sys)$)-  as follows
\begin{itemize}
\item To each straight line between the points $(t_i,x_i,l_i)$ and $(t_{i+1},x_{i+1},l_{i+1})$, one associates the operators  $\caU_{t_{i+1}-t_i}$.
\item To each point $(t_i,x_i,l_i)$, one associates the operator $ \la^2 \caI_{x_i,l_i}$, defined in \eqref{def: operatorI}.
\item To each curved line between  the points $(t_r,x_r,l_r)$ and $(t_s,x_s,l_s)$, with $r<s$, we associate the factor 
\beq 
  \delta_{x_r,x_s}   \left\{  \begin{array}{cc}   \hat\psi (t_s-t_r )   & l_{r}=\links \\ [1mm]  
 \hat\psi (-(t_s-t_r ))    & l_{r}=\rechts.
  \end{array}\right. \nonumber
\eeq
\end{itemize}
Rules like these are commonly called "Feynman rules" by physicists.
}
 \end{minipage}

\end{figure}

We introduce some more terminology,  extending the above definition of pairings. It will be helpful in  classifying  the pairings.

\begin{definition}\label{def: pairings}
\ben
\item{Let $\Si_n$ be  the set of sets of $n$ pairs of (distinct) natural numbers. More concretely, for each  $\dsi \in \Si_n$, we can write
\beq\label{not: pairs}
\dsi= \left\{  (r_1,s_1), \ldots,  (r_n,s_n) \right\},   r_i , s_i  \in \bbN,  
\eeq
for natural numbers $r_i,s_i, i=1,\ldots,n$ which are all distinct. By convention,  $r_i<s_i, i=1,\ldots,n$ and $r_i< r_{i+1}, i=1,\ldots,n-1 $. If $\si_1  \in \Si_{n_1}$ and $ \si_2   \in \Si_{n_2}$, we write $\si_1 < \si_2$ whenever all elements of the pairs $(r_i^1,s_i^1)$ in $\si_1$  are smaller than all elements of the pairs $(r_j^2,s_j^2)$ in $\si_2$,  i.e.,
\beq
  s_i^1 < r_j^2, \qquad  i=1,\ldots, n_1,  j=1,\ldots, n_2 .  
\eeq 
}

\item{    Recall the definition of  $\caP_n$, the set of pairings with $n$ pairs. Obviously $ \caP_n \subset \Si_n$, and $\si \in \Si_n$ belongs to $\caP_n$ whenever  $\cup_{i=1}^n \{ r_i,s_i\} =\{1,\ldots,2n \}$ .
  Further, with any  $\si \in \Si_n $, we associate  the unique pairing $\pi\in \caP_{n}$ for which there is a monotone increasing function $q$ on $\{1,\ldots, 2 n \}$ such that 
\beq\label{connection si and pi}
(i,j) \in \pi \Leftrightarrow (q(i), q(j)) \in \si.  \eeq}
\item{
We set $\caP:= \cup_{n \geq 1}  \caP_n  $ and write $\str \pi \str=n$ whenever $\pi \in \caP_n$. 
   }
   \item{ We call  $\si \in \Si_n$  
  irreducible (Notation: $\mathrm{irr.\ }$) whenever there are no two sets $\si_1  \in \Si_{n_1}, \si_2   \in \Si_{n_2}, n_1+n_2=n$ such that  $\si=\si_1 \cup \si_2$ and $\si_1 < \si_2$. For any  $\si \in \Si_n$ that is not irreducible, we can thus find   partitioning subsets $\si_1,\ldots, \si_m$ ($\cup_{i=1}^m \si_i = \si$) such that $\si_{i=1,\ldots,m}$ are irreducible and $\si_i < \si_{i+1}$ for $i=1,\ldots,m-1$.   
  } 
  \item{ Consider some $\pi \in \caP$ and its partitioning into irreducible subsets $\si_1,\ldots,\si_m$, as defined above.  By \eqref{connection si and pi}, we can associate to each of the $\si_i$ a unique $\pi_i$ in $\caP$.  We call the set $(\pi_1,\ldots,\pi_m)$ of pairings, obtained in this way the  decomposition of $\pi$ into irreducible components.
}
  
\item{  For each $n \in \bbN$, we define a distinguished pairing $\pi \in \caP_n$, which is called the  \emph{minimally irreducible} pairing  (Notation: $\mathrm{min. irr.\ }$).  For $n>2$, this  minimally irreducible pairing is  given by
\beq
(r_1,s_1)=(1,3) , \quad   (r_{n},s_n)= (2n-2,2n), \quad   (r_{i+1},s_{i+1}) = (2i, 2i+3),  \quad \textrm{for} \quad i=1,\ldots, n-2.
\eeq
 For $n=1$ and $n=2$, the minimally irreducible pairing is defined to be $(1,2)$  and $\{(1,3),(2,4)\}$ respectively.   Intuitively, the minimally irreducible pairing in $\caP_n$ is characterized by the fact that if one  removes any pair, other than the pair with $r=1$ or $s=2n$, the resulting pairing is  no longer irreducible. }
\een
\end{definition}
\vspace{4mm}

\begin{figure}[ht!]
\begin{center}
\psfrag{n1}{$1$}
\psfrag{n2}{$2$}
\psfrag{n3}{$3$}
\psfrag{n4}{$4$}
\psfrag{n5}{$5$}
\psfrag{n6}{$\ldots$}
\includegraphics[width = 10cm]{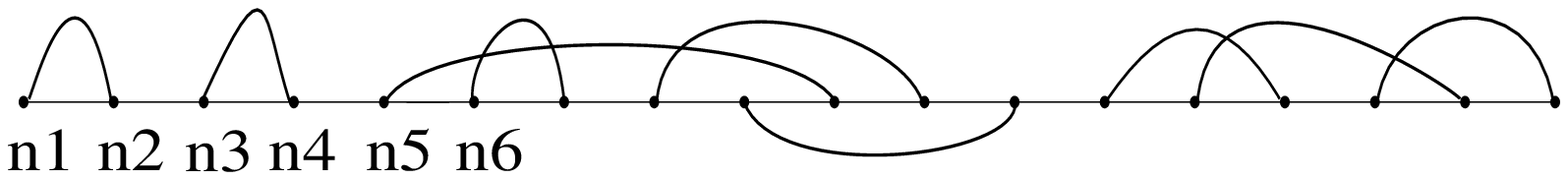}
\end{center}
\caption{\footnotesize{ Graphical representation of a pairing $\pi \in \caP_9$.  The pair $(r,s)$ belongs to $\pi$ whenever the natural numbers $r,s$ are connected by an arc.  This type of diagrams differs from those of Figure \ref{fig: explanation V}  in that we don't keep track of the $t_i$-coordinates, but only of the topological structure of the pairings.
Below is the decomposition of $\pi$ into irreducible components.}  \label{fig: decompose} } 
\begin{center}
\psfrag{\pair1}{$\pi_1$}
\psfrag{\pair2}{$\pi_2$}
\psfrag{\pair3}{$\pi_3$}
\psfrag{\pair4}{$\pi_4$}
\includegraphics[width = 15cm]{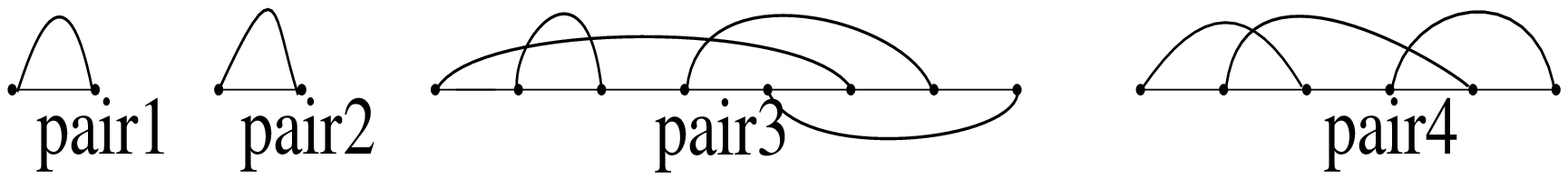}
\end{center}
\caption{\footnotesize{The irreducible components $\pi_1,\pi_2,\pi_3,\pi_4$. Explicitly, $\pi_1=\pi_2= \{(1,2)\}$, $\pi_3= \{  (1,6), (2,3), (4,7), (5,8) \}$ and $\pi_4=\{(1,3), (2,5), (4,6)\}$. The pairings $\pi_1,\pi_2$ and $\pi_4$ are minimally irreducible, whereas $\pi_3$ is not.  Indeed, one can remove the pair $(4,7)$ from $\pi_3$ without destroying the irreducibility.} } 
\end{figure}

\noindent For an irreducible pairing $\pi \in \caP_n$, we introduce (using the same conventions as in \eqref{eq: Z with pairings}, \eqref{def: zeta}),
\beq\label{def: vpairing}
\caV_{t}(\pi) :=    \mathop{\int}\limits_{0 = t_1\leq \ldots \leq t_{2n} = t}  (\prod_{i=2}^{2 n-1}\d t_i)  \mathop{\sum}\limits_{\underline{x},\underline{l}}  \,     \zeta_\pi(\underline{t},\underline{x}, \underline{l} )    \,     \caI_{x_{2n},l_{2n} } \caU_{t-t_{2n-1}}   \ldots \caI_{x_2,l_2}   \caU_{t_2-t_1} \caI_{x_1,l_1} .
\eeq
We can now rewrite \eqref{eq: Z with pairings}  as a sum over collections of irreducible pairings;
\baq \label{eq: z as polymer model}
&& \caZ^{\la}_{t}  =   \sum_{m \in \bbZ^+}  \mathop{\int}\limits_{0 \leq t_1 \ldots \leq t_{2m} \leq t} \d t_1 \ldots \d t_{2m} \nonumber \\[2mm]
&& \qquad        \mathop{\sum}\limits_{\scriptsize{\left.\begin{array}{c}  \pi_1, \ldots, \pi_{m}   \in \caP   \\ \pi_1, \ldots, \pi_{m} \,  \mathrm{irr.\ }   \end{array}\right.} }         \la^{(2 \sum_{i=1}^{m}\str \pi_i \str)}         \caU_{t-t_{2m}} \caV_{t_{2m}-t_{2m-1}}(\pi_m)   \ldots        \caU_{t_3-t_2} \caV_{t_2-t_1}(\pi_1)  \caU_{t_1}. 
\eaq
To obtain this last expression, we decompose each pairing $\pi$ in \eqref{eq: Z with pairings} into its irreducible components $\pi_1,\ldots, \pi_m$, and we made use of a simple factorization property of \eqref{eq: Z with pairings}.  The term on the RHS of \eqref{eq: z as polymer model} corresponding to $m=0$ is understood to be equal to $\caU_t$.

In expression \eqref{eq: z as polymer model}, we view  the pairings $\pi_i$   with $\str\pi_i\str \geq 2$ as excitations.  If  $\str\pi_i\str=1$,  for all $i=1,\ldots,m$, the corresponding term in \eqref{eq: z as polymer model} is  called a \emph{ladder} diagram. These ladder diagrams  provide the leading contribution to the dynamics, and they are the only terms that survive in the kinetic limit.
We define separately the Laplace transforms of the  irreducible  "excitation" diagrams ($\caR^{\mathrm{ex}}_\la$) and the irreducible ''ladder" diagram ($\caL$):
\baq
\caR^{\mathrm{ex}}_\la (z) &:=& \int_{\bbR^+}  \d t \,    \e^{-t z }    \mathop{\sum}\limits_{\scriptsize{\left.\begin{array}{c}  \str\pi\str\geq 2 \\ \pi \,\,  \mathrm{irr.\ } \end{array}\right.} }    \la^{2 \str \pi \str}    \caV_{t}(\pi)  \label{def: resolvent exitations}  \\
\caL (z) &:=& \int_{\bbR^+}  \d t \,   \e^{-t z }   \sum_{\str \pi \str =1}  \caV_{t}(\pi) =   \int_{\bbR^+}  \d t \,     \e^{-t z }   \caV_{t}(\{(1,2)\}) .  \label{def: L}
\eaq
Here and in what follows, we omit the specification $\pi \in \caP$ under the summation symbol.  We observe that, 
in \eqref{def: L},  the only element of $\caP_1$ is the set containing the single pair $(1,2)$.  The operators $\caR^{\mathrm{ex}}_\la (z)$ and $\caL (z) $ have already appeared in Theorem \ref{thm: polymer model}.    We will prove  Theorem \ref{thm: polymer model} in Section \ref{sec: proof of polymer model}. First, we establish some helpful estimates

\subsection{Estimates on the Dyson expansion}\label{sec:  estimates on dyson expansion}

\subsubsection{A priori estimates}
The following Lemma \ref{lem: apriori estimate} is a useful a-priori estimate. Its main assertion, Statement 2), i.e., eq.\  \eqref{awful estimate}, gives a bound on $\caV_{t}(\pi)$, the contribution of the irreducible pairing $\pi$ to the dynamics, in terms of the temporal coordinates $\underline{t}$. In particular, the sum over the other coordinates, $\underline{x}$ and $\underline{l}$ is already  performed. This is possible because the matrix elements of the free dynamics $(\e^{-\i t H_\sys})(0,x)$ decay exponentially in space, for fixed $t$; (see Statement 1 of Lemma   \ref{lem: apriori estimate}, or eq.\ \eqref{eq: decay of space matrix elements}).  Eq.\ \eqref{eq: bound1} tells us  that one can sum over $x$ at the cost of introducing an exponential growth in time.    This  exponential growth in time  is also visible in \eqref{awful estimate}, in the factor $\e^{ 2t c_{\systemdispersion}(\ga_1) } $.  However, this exponential growth is harmless, because the reservoir correlation functions $\hat \psi$ on the RHS of \eqref{awful estimate} are exponentially decaying in time, by Assumption \ref{ass: exponential decay}, and the growth constant $c_{\systemdispersion}(\ga_1)$ can be chosen arbitrarily small. In particular, it can be chosen smaller than the reservoir decay rate $g_\res$, and this fact will be exploited in Section \ref{sec: properties of Rex}.
.

\begin{lemma}\label{lem: apriori estimate}
Suppose that Assumption \ref{ass: analytic dispersion} holds (with some $\de_{\systemdispersion}>0$) and define 
\[ \left. \begin{array}{lcll}
 c_{\systemdispersion}(\delta)&:=& \sup_{k \in \tor} \sup_{\str \Im \ka \str \leq \delta} \str\Im\systemdispersion(k+\ka)\str ,  &     \left(c_{\systemdispersion}(\delta) <\infty, \quad   \textrm{for}\quad 0<\delta<\delta_{\systemdispersion} \right) \\[2mm]
 b_d(\delta) &:=& \sum_{x \in \lat} \e^{-\delta \str x \str},    \qquad &     \left(b_d(\delta)<\infty, \quad   \textrm{for}\quad 0<\delta \right).
\end{array} \right.  \]
Then the following statements hold true. 
\ben
\item{  For any $\ka \in \bbC^d$ with $ \str \Im \ka \str \leq \ga_1$, for some $\ga_1<  \delta_{\systemdispersion}$,
\beq \label{eq: bound on unitaries}
\norm   \e^{\i (\ka, X) } \e^{-\i t \systemdispersion(P)}   \e^{-\i (\ka, X) }  \norm \leq      \e^{t c_{\systemdispersion} ( \ga_1 )} , \qquad  t \geq 0 . 
\eeq
}
\item{ Let $\pi \in \caP_n$,  and choose constants $0<\ga < \ga_1 < \delta_{\systemdispersion}$.  For any $\ka \in \bbC^d$ satisfying $ \str\Im \ka \str \leq  \ga_1-\ga$, 
\beq \label{awful estimate}
\left\norm   \caJ_{\ka} \, \caV_{t}(\pi)\,  \caJ_{- \ka}   \right\norm \leq   \left\{ \begin{array}{l}    b_d(2\ga) \,  [ b_d(\ga_1-\ga-\str\Im \ka \str)]^{2 n} \,  2^{2n}  \,  \e^{ 2t c_{\systemdispersion}(\ga_1) }  \\[2.5mm]
  \times \,   \mathop{\int}\limits_{0 =t_1\leq  \ldots \leq t_{2n} = t} \left( \mathop{\prod}\limits_{i=2}^{2 n-1}\d t_{i}  \right)  \mathop{\prod}\limits_{(r,s) \in \pi} \str \psi(t_s-t_r)\str,  \end{array}\right. 
\eeq
We recall that  $\norm \cdot \norm$  in \eqref{awful estimate} refers to the operator norm on  $\scrB(\scrB_2(\scrH_\sys))$.
}
\een

\end{lemma}

\begin{proof}
\textbf{Statement 1)} \\
Recall that $H_\sys=\systemdispersion(P)$.  By analytic continuation from $\Im \ka=0$ to $\str \Im \ka \str \leq \delta_\systemdispersion$, one has that
\beq
 \e^{\i (\ka, X) }  \e^{-\i  t \systemdispersion(P)}    \e^{-\i (\ka, X) } =   \e^{ -\i   t   \systemdispersion(   P -    \ka   )    }  .
\eeq
 Since, for $\str \Im \ka \str \leq \ga_1$
\beq 
 \norm \e^{ -\i  t   \systemdispersion(   P-  \ka  )    }   \norm  \leq   \e^{ t \norm \Im \systemdispersion(   P-  \ka  )  \norm }   \leq    \e^{t c_{\systemdispersion}(\ga_1)},  \qquad  t \geq 0,
\eeq
the claim \eqref{eq: bound on unitaries} is proven.
We observe that \eqref{eq: bound on unitaries} implies
 \beq \label{eq: decay of space matrix elements}
 \str (\e^{-\i t H_\sys})(x,x') \str \leq \e^{t c_{\systemdispersion} (\ga_1)}  \e^{-\ga_1 \str x'-x \str}, \qquad  \textrm{for any} \, \, 0<\ga_1 < \delta_\systemdispersion , \qquad  t \geq 0,
  \eeq
  and hence 
  \beq \label{eq: bound1}
\sum_{x' \in \lat}   \e^{\ga \str x'-x \str}  \str (\e^{-\i t H_\sys})(x,x') \str \leq \e^{t c_{\systemdispersion} (\ga_1)}   b_d(\ga_1-\ga) , \quad  \textrm{for any} \,\,  0<\ga< \ga_1<\delta_\systemdispersion , \qquad  t \geq 0 .
  \eeq

 \textbf{Statement 2)} \\
 To estimate the integrand in \eqref{def: vpairing}, we choose  $0<\ga'< \ga_1<\delta_\systemdispersion$ and find that
  \baq
  && \sum_{y', z'} \left\str \e^{\ga' (\str y'-y\str + \str z'-z \str )}  \left( \mathop{\sum}\limits_{\underline{x}, \underline{l}}     \zeta_\pi(\underline{t},\underline{x}, \underline{l} )            \caI_{x_{2n},l_{2n} }    \ldots \caI_{x_2,l_2}   \caU_{t_2-t_1} \caI_{x_1,l_1}   \right)(y,z;y',z') \right\str \nonumber \\[2mm]
    &\leq&   (\sup_{\underline{x}, \underline{l} } \str  \zeta_\pi(\underline{t},\underline{x}, \underline{l} ) \str )    \sum_{\underline{l}}   \e^{2 t c_{\systemdispersion}(\ga_1)}     (  b_d(\ga_1-\ga'))^{2n} , \label{eq: iterated exponential bound}
\eaq
where we can replace "$\sum_{\underline{l}}$" by $2^{2n}$, the number of terms in the sum.  
The bound  \eqref{eq: iterated exponential bound} is obtained by applying  \eqref{eq: bound1} $2 n$ times.

For clarity, we illustrate this with an example: Take $n=4$ and $(l_1, l_2, l_3,l_4, l_5, l_6,l_7, l_8)= (\links, \rechts, \links, \links, \rechts, \links, \rechts, \rechts) $. First, we notice that 
\beq
 \left\str \left( \caI_{x_{8},l_{8} }    \ldots \caI_{x_2,l_2}   \caU_{t_2-t_1} \caI_{x_1,l_1}   \right)(y,z;y',z')\right\str    
\eeq
 vanishes unless $x_1=y$ and $x_8=z'$, and that it is bounded by
\beq
 \left\{
\begin{array}{l}  w (t_3-t_1,x_3-x_1) \times  w (t_4-t_3,x_4-x_3)   \times  w(t_6-t_4,x_6-x_4)    \times  w(t-t_6,  y'-x_6)     \\[1mm]
     \times \,\,   w(t_2-0,x_2-z)  \times  w(t_5-t_2,x_5-x_2) \times  w(t_7-t_5,x_7-x_5)   \times  w(t_8-t_7,x_8-x_7)  ,  \end{array}\right. \label{eq: ridiculous}
\eeq
where $w(u,x):=  \str (\e^{-\i u H_\sys } )(0,x)  \str, t_1=0, t_8=t$.

We use the decomposition (recall that $x_1=y$ and $x_8=z$)  \baq  &\str y'-y\str  \leq  \str x_3- x_1 \str  + \str x_4- x_3 \str  + \str x_6-x_4 \str + \str y' -x_6 \str  \nonumber,   \\
  &\str z'-z \str \leq   \str x_2 -z \str  + \str x_5- x_2 \str  + \str x_7-x_5 \str + \str x_8 -x_7 \str \nonumber,  \eaq 
  and \eqref{eq: ridiculous} to factorize the sum over $y',z', \underline{x}$ on the LHS of \eqref{eq: iterated exponential bound}. Those sums can then be carried out with the help of \eqref{eq: bound1}, yielding the bound
\baq
&&(b_d(\ga_1-\ga'))^{8}   \left\{
\begin{array}{l}  \exp{ \left( c_{\systemdispersion}(\ga_1) \left[   (t_8-t_6) +  (t_6-t_4)+  (t_4-t_3)+  (t_3-t_1)  \right] \right)  }  \\[1mm] \times   \exp{  \left(  c_{\systemdispersion}(\ga_1) \left[      (t_8-t_7) +  (t_7-t_5)+  (t_5-t_2)+  (t_2-0)  \right]  \right)  }   
 \end{array}\right\}   \nonumber \\[2mm]
&& \qquad \qquad  =   (b_d(\ga_1-\ga'))^{8}  \,   \e^{2 t c_{\systemdispersion}(\ga_1)}
\eaq
Note that this bound only depends on  $\str \pi \str$ and $t$, and not on $\underline{t},\underline{l}$, or $\pi$. Hence it can  be applied for all  $\underline{l}$, which yields the factor $2^{2n}$ in \eqref{eq: iterated exponential bound}.
\vspace{0.5cm}

For a linear operator $\caW$ on $l^2(\bbZ^d\times \bbZ^d)$, a straightforward application of the  Cauchy-Schwarz inequality yields
\beq\label{eq: cs bound}
 \norm \caW \norm \leq   b_d(2\de) \left(\sup_{y,z \in \bbZ^d}  \sum_{y' ,z' \in \bbZ^d}  \str \caW(y,z; y' ,z' )\str   \e^{\delta (\str y'-y \str + \str z'-z \str )} \right).
\eeq
Starting from the explicit definition of $\caJ_{\ka} \, \caV_{t}(\pi)\,  \caJ_{- \ka}  $ (as in \eqref{def: jkappa} and \eqref{def: vpairing}) , one uses \eqref{eq: cs bound} and \eqref{eq: iterated exponential bound} with $\ga':= \ga+\str \Im \ka\str$. This yields Statement 2).

\end{proof}

\subsubsection{A combinatorial estimate}

In the next step of our analysis of the Dyson series, we  show that one can perform the integration over all pairings $\pi$ and temporal coordinates $\underline{t}$ contributing to  \eqref{def: resolvent exitations}.  
The following lemma is purely combinatorial, i.e.,  it only employs notions introduced in Definition \ref{def: pairings}.

\begin{lemma}\label{lem: combinatorics}
Consider a positive function $h$ on $\bbR^+$ and a pairing $\pi \in \caP$.  We define
\beq
\chi_t(\pi):= \mathop{\int}\limits_{0 = t_1\leq \ldots \leq t_{2n} = t}  (  \prod_{i=2}^{2n-1}  \d t_i)    \prod_{(r,s) \in \pi}    h(t_s-t_r) , \qquad  \textrm{with}\, n=\str \pi \str.
\eeq
Then
\beq
\label{eq: sum over irreducible}
\mathop{\sum}\limits_{ \pi \,  \mathrm{ irr.\ } }  \chi_t(\pi)  \leq     \Big( \mathop{\sum}\limits_{\pi \,  \mathrm{min.\ irr.\ }  } \chi_t(\pi) \Big) \times \exp{\left(t \int_{\bbR^+} \d w h(w) \right)},
\eeq
and, if  $\pi$ is the minimally irreducible pairing in $\caP_n$ and $z \in \bbR$,
\beq
 \int_{\bbR^+} \d t \,  \e^{-tz}     \chi_t(\pi) 
  \leq    \left( \int_{\bbR^+} \d w h(w) \e^{-w z}\right)\times  \left(\int_{\bbR^+} \d y \int_{\bbR^+} \d w \,  h(y+w)  \e^{-w z} \right)^{n-1} . \label{eq: calculation excitation pairings}
\eeq

\end{lemma}

\begin{proof}
Given $\pi \in \caP_n$, we can relabel the times $t_1,\ldots,t_{2n}$ by setting
\beq\label{change of variables}
u_{i}= t_{r_i}, \qquad    v_{i}= t_{s_i}   \qquad  \textrm{for} \quad i=1,\ldots,n.
\eeq
Using our conventions for the labels of the pairs $(r_i,s_i)$, it follows that 
\beq 0\leq u_i \leq v_i \leq t ,\qquad 0\leq u_i  \leq u_{i+1}\leq t, \qquad  0=u_1, t=\max\{v_i \} .  \label{conditions on uv}\eeq
Conversely, a set of $n$ pairs of times $(u_i,v_i) ,i=1,\ldots,n$, satisfying \eqref{conditions on uv} uniquely determines a pairing $\pi \in \caP_n$ and corresponding times $0=t_1\leq \ldots \leq t_{2n}=t$. 
\begin{figure}[ht!] 
\vspace{0.5cm}
\begin{center}
\psfrag{u1}{$u_1$}
\psfrag{u2}{$u_2$}
\psfrag{u3}{$u_3$}
\psfrag{v1}{$v_1$}
\psfrag{v2}{$v_2$}
\psfrag{v3}{$v_3$}
\psfrag{t1}{$t_1$}
\psfrag{t2}{$t_2$}
\psfrag{t3}{$t_3$}
\psfrag{t4}{$t_4$}
\psfrag{t5}{$t_5$}
\psfrag{t6}{$t_6$}
\includegraphics[width = 4cm, height=2cm]{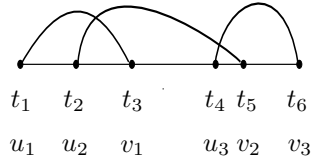}
\end{center}
\caption{ \footnotesize{ This figure illustrates the change of variables from $(\pi,\underline{t})$, with $\pi \in \caP_3$ and $t_1 < \ldots <t_6$, to $(u_i,v_i)_{i=1,2,3}$, with $u_i \leq v_i$ and $u_i \leq u_{i+1}$. }
 \label{fig: single graph with times} } 
\end{figure}

Consider an irreducible pairing $\pi' \in \caP_{n'}$. It is easy to see that we can always find a subset $j_1,\ldots,j_{n}$  of $\{ 1,\ldots,n'\}$, for some $n\leq n'$, such that 
\ben
\item{the pairs $(r_{j_i}, s_{j_i}), i=1,\ldots,n$ determine a minimally irreducible pairing $\pi \in \caP_n$;}
\item{these pairs contain the boundary points, i.e.\ $r_{j_1}=1$  and $\max_{i}\{s_{j_i} \}=2n' $.}
\een
  We write $\pi'  \rightarrow \pi$  whenever $\pi$  and $\pi'$ are related in this way; (note, however, that $\pi$ is not uniquely determined). 
  It follows that 
  \beq \label{eq: inequality nonunique pairings}
        \mathop{\sum}\limits_{\scriptsize{\left.\begin{array}{c}  \str\pi' \str= n' \\ \pi' \,  \mathrm{irr.\ } \end{array}\right.} }      \chi_t(\pi')         \leq        \mathop{\sum}\limits_{\scriptsize{\left.\begin{array}{c}  \str\pi \str\leq n' \\ \pi \,  \mathrm{min.\. irr.\ } \end{array}\right.} }  
          \mathop{\sum}\limits_{\scriptsize{\left.\begin{array}{c}  \str \pi'\str = n'    \\      \pi'  \rightarrow \pi  \end{array}\right.} }  
       \chi_t(\pi')  .
  \eeq
For $n' \geq 2$, the inequality is strict, since $\pi$ is not necessarily uniquely determined by $\pi'$, and hence the same irreducible $\pi'$  can appear more than once on the RHS of \eqref{eq: inequality nonunique pairings}.\\

Using the change of variables \eqref{change of variables}, one can convince oneself  that, for all $m:=n'-n \geq 0$, 
\beq \label{eq: adding pairings to min irr}
       \mathop{\sum}\limits_{\scriptsize{\left.\begin{array}{c}  \str \pi'\str = n'    \\      \pi'  \rightarrow \pi \end{array}\right.} }   \chi_t(\pi') =
\chi_t(\pi)   \mathop{\int}\limits_{ \scriptsize{\left.\begin{array}{c} 0 \leq u_1\leq\ldots\leq u_{m} \leq t \\  0 \leq u_i \leq v_i \leq t   \end{array}   \right.}}  \d \underline{u}\d \underline{v}\,  \prod_{i=1}^{m} h(v_i-u_i) ,  \eeq
where  $\pi$ is the minimally irreducible pairing in $  \caP_n$, and where we  have abbreviated $\d \underline{u}:= \d u_1 \ldots \d u_n$ and $\d \underline{v}:= \d v_1 \ldots \d v_n$.
The relation \eqref{eq: adding pairings to min irr} expresses the fact that  one can add any set of pairs, corresponding to times $\underline{u},\underline{v}$ satisfying the first two conditions of \eqref{conditions on uv}, to a minimally irreducible $\pi$, thus obtaining a new irreducible pairing (see also Figure \ref{fig: nonminimal pairs}).
\begin{figure}[ht!] 
\vspace{0.5cm}
\begin{center}
\psfrag{ua1}{$u_1$}
\psfrag{ua2}{$u_2$}
\psfrag{va1}{$v_1$}
\psfrag{va2}{$v_2$}
\psfrag{ub1}{$u_1$}
\psfrag{ub2}{$u_2$}
\psfrag{vb1}{$v_1$}
\psfrag{vb2}{$v_2$}
\psfrag{uc1}{$u_1$}
\psfrag{uc2}{$u_2$}
\psfrag{vc1}{$v_1$}
\psfrag{vc2}{$v_2$}
\includegraphics[width = 15cm]{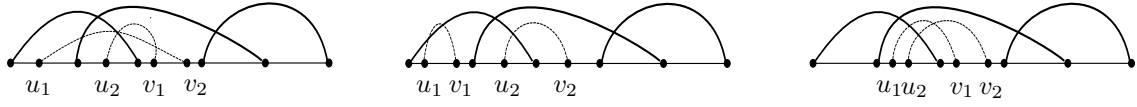}
\end{center}
\caption{ \footnotesize{Illustration of \eqref{eq: adding pairings to min irr}.  Three pairings in $\caP_5$ contributing to the LHS of \eqref{eq: adding pairings to min irr}. We have chosen $\pi$ to be the minimally irreducible pairing $(1,3), (2,5), (4,6)$ in $\caP_3$, as in Figure \ref{fig: single graph with times} . For each of these $3$ pairings in $\caP_5$, the five pairs  $(u_i,v_i)_{i=1,\ldots,5}$ contain a subset of three pairs identified with $\pi$. We have only shown the two other pairs, relabeling them as $(u_i,v_i)_{i=1,2}$.  The same strategy is used to prove \eqref{eq: adding pairings to min irr} in general. }   \label{fig: nonminimal pairs} } 
\end{figure}

By explicit computation, 
\baq
     \sum_{m \in \bbZ^+}    \mathop{\int}\limits_{ \scriptsize{\left.\begin{array}{c} 0 \leq u_1\leq\ldots\leq u_m \leq t \\  0\leq u_i \leq v_i \leq t   \end{array}   \right.}} \d \underline{u}\d \underline{v}  \prod_{i=1}^m  h(v_i-u_i)  
    & \leq &     \sum_{m \in \bbZ^+}      \mathop{\int}\limits_{0 \leq u_1\leq\ldots\leq u_m \leq t}  \d \underline{u}       \left( \int_{\bbR^+} \d w h(w) \right)^m  \nonumber \\
    &\leq &   \exp{\left(t \int_{\bbR^+} \d w h(w) \right)},
\eaq
which proves the bound \eqref{eq: sum over irreducible} starting from  \eqref{eq: inequality nonunique pairings} and \eqref{eq: adding pairings to min irr}.

 When we perform the change of variables \eqref{change of variables} for  a minimally irreducible pairing $\pi$, the variables $\underline{u}, \underline{v}$ satisfy the constraint $ u_{i+1} \leq v_i \leq u_{i+2}$ in addition to the constraints $0\leq u_i  \leq u_{i+1} \leq t$ and $ 0\leq  u_i \leq v_i \leq t$.  
 Let $\pi$ be the minimally irreducible pairing in $\caP_n$.  Then ($u_1=0$ is a dummy variable) 
\baq
&&     \int_{\bbR^+} \d t  \e^{-tz}    \chi_t(\pi)   =   \int_0^\infty   \d v_1h(v_1-u_1)   \e^{-z(v_1-u_1)}  \,       \int_0^{v_{1} } \d u_{2}  \int_{v_{1}}^{\infty } \d v_{2}    \ldots     \nonumber  \\[1mm]     
   &  &     \qquad \ldots \qquad    \int_{v_{n-5}}^{v_{n-4} } \d u_{n-2}     \int_{v_{n-3}}^{\infty } \d v_{n-2}       \ldots     \int_{v_{n-3}}^{v_{n-2} }  \d u_{n-1}  \int_{v_{n-2}}^{\infty } \d v_{n-1}    \e^{-z(v_{n-1}-v_{n-2})}   h(v_{n-1}-u_{n-1})       \nonumber    \\[1mm]    
   && \quad \qquad \qquad       \int_{v_{n-2}}^{v_{n-1} } \d u_n   \int_{v_{n-1}}^{\infty } \d v_n     \e^{-z(v_n-v_{n-1})}   h(v_n-u_n)   .  
\eaq
Performing the change of variables $w_i=v_i-v_{i-1}$ and $y_i=v_{i-1}-u_i $ (for $i>1$) and extending the range of integration of $y_i$ to $\bbR$, the above expression factorizes and one obtains  the bound \eqref{eq: calculation excitation pairings}.

\end{proof}

\subsection{Proof of Theorem \ref{thm: polymer model}}\label{sec: proof of polymer model}

In this section, we prove Theorem \ref{thm: polymer model}.  Statement 2) is proven separately for $\caL(z)$ and $\caR^{\mathrm{ex}}_\la(z)$ in Sections \ref{sec: properties of L} and \ref{sec: properties of Rex}, respectively. Statement 3) is proven in Section \ref{sec: properties of L} and Statement 1) in Section \ref{sec: proof of resolvent formula}.

It is mainly in Section \ref{sec: properties of Rex} that we use the preparatory work summarized in Lemma \ref{lem: apriori estimate} and Lemma \ref{lem: combinatorics}, in order  to obtain a bound on $\caR^{\mathrm{ex}}_\la (z) $. 

\subsubsection{Properties of $\caL(z)$}\label{sec: properties of L}

We compute $\caL(z)$ starting from \eqref{def: vpairing} and \eqref{def: L}:
\beq \label{eq: L as second order perturbation}
\caL(z)=    \int_{\bbR^+} \d t \,  \e^{-t z}    \,   \sum_{x \in \lat} \sum_{l,l' \in \{\links, \rechts \}}  \,           \caI_{l',x}        \caU_{t}   \caI_{l,x}    \left\{  \begin{array}{cc}   \hat\psi ( t)   & l=\links \\ [3mm]  
 \hat\psi (- t)  & l=\rechts
  \end{array}\right.
\eeq
To display the result, we introduce  the functions $\psi_+,\psi_-$ as
\beq
\psi_{+}(z)=  \int_{\bbR^{+}}  \d t \hat\psi(t) \e^{ \i t z }   , \qquad  \psi_{-}(z)=  \int_{\bbR^{-}}  \d t \hat\psi(t) \e^{ \i t z }, \qquad  z \in \bbC,
\eeq
with $\hat \psi$ as  defined in Section \ref{sec: full reservoir}; (we recall that $\hat \psi(u)$ decays exponentially).
Since $\hat \psi(-u)= \overline{\hat \psi(u)} $ (as follows from \eqref{def: correlation function}), one has that
 \beq \label{prop of psi plus and minus} \psi_{+}(z)= \overline{\psi_{-}(\bar z) } , \qquad  \psi(z)= \psi_{+} (z)+ \psi_{-} (z) , \qquad  \textrm{with}\,  \str \Im z \str < g_{\res}. \eeq
Using \eqref{eq: L as second order perturbation}, we calculate $\caL(z)S $, for $S \in L^2(\tor \times \tor)$,
\beq \label{eq: explicit caL} \begin{split}
&(\caL(z) S)(k+\frac{p}{2},k-\frac{p}{2})  \\
  &   = \int_{\bbT^d} \d k'     \,  \left( \psi_+ [ \systemdispersion(k-\frac{p}{2}) -\systemdispersion(k'+\frac{p}{2})+ \i z ]+ \psi_- [ \systemdispersion(k+\frac{p}{2})-\systemdispersion(k'-\frac{p}{2}) -\i z]   \right)              S (k'+\frac{p}{2},k'-\frac{p}{2})  \\
  &-              \int_{\bbT^d} \d k' \,   \left( \psi_+ [ \systemdispersion(k-\frac{p}{2})-\systemdispersion(k'+\frac{p}{2})+\i z) ]   +  \psi_-[  \systemdispersion(k+\frac{p}{2}) - \systemdispersion(k'-\frac{p}{2}) - \i z ]  \right)      S(k+\frac{p}{2},k-\frac{p}{2}).
  \end{split}
\eeq

 The claim about $\caL(z)$ in Statement 2) of Theorem \ref{thm: polymer model} follows by noticing that the above expression can be analytically continued in $z$ and $p$. This follows from the analyticity of $\systemdispersion$ (Assumption \ref{ass: analytic dispersion}) and $\psi_+, \psi_-$ (consequences of Assumption \ref{ass: exponential decay}).  
 
 To prove Statement 3), we first check that $(\caL(0))_0=M^0$ by setting $p=0$ and $z=0$ in \eqref{eq: explicit caL},  and using \eqref{prop of psi plus and minus}.  
 It remains to verify that 
 \beq\label{equality of hamiltonian part}
 \la^2(M^\ka-M^0) =   \i  \la^2 (\ka, \nabla \systemdispersion)  = (\ad(\i H_\sys))_{\la^2 \ka} +O((\la^2 \ka)^2) 
 \eeq
  as operators on $L^2(\tor) $, where $(\ka, \nabla \systemdispersion)$ is the multiplication operator given by the function  $k \mapsto (\ka, \nabla \systemdispersion)(k)$.  Equation \eqref{equality of hamiltonian part} follows by writing explicitly
  \beq
  ((\ad(\i H_\sys))_p \theta)(k)= \i  \left(\systemdispersion (k+\frac{p}{2})    -  \systemdispersion(k-\frac{p}{2})  \right) \theta(k), \qquad \theta \in L^2(\tor, \d k),
  \eeq
expanding in powers of $p$ and putting $p=\la^2 \ka$.

\subsubsection{Properties of  $\caR^{\mathrm{ex}}_\la(z)$}\label{sec: properties of Rex}

Choose positive constants $\ga_1>\ga>0$, as in Lemma \ref{lem: apriori estimate}, and define  the quantity $\chi_t(\pi)$ as in Lemma \ref{lem: combinatorics}, with  $h$ given by
\beq
h(t):= 2^2  b_d(\ga_1-\ga-\str \Im \ka \str)^2 \la^2  \str \hat\psi(t)\str.
\eeq
It follows from  Statement 2) of Lemma \ref{lem: apriori estimate} and equations \eqref{def: vpairing}, \eqref{def: resolvent exitations} that
\beq \label{eq 1 of proof excitations}
\norm \caJ_{\ka} \caR^{\mathrm{ex}}_\la (z)  \caJ_{-\ka} \norm    \leq b_d(2\ga)   \int_{\bbR^+} \d t  \,  \e^{ 2 c_{\systemdispersion}(\ga_1)t}  \e^{- \Re z t} \, \Bigg( 
  \mathop{\sum}\limits_{\scriptsize{\left.\begin{array}{c}\str\pi\str >1\\ \pi \,  \mathrm{irr.\ } \end{array}\right.} } 
    \chi_t(\pi)\Bigg) .
\eeq
and hence, using Lemma \ref{lem: combinatorics}, that
\baq
\norm \caJ_{\ka} \caR^{\mathrm{ex}}_\la (z)  \caJ_{-\ka} \norm   
&\leq &  b_d(2\ga) \int_{\bbR^+} \d t \e^{-(\Re{z}-a) t}  \sum_{\pi \,   \mathrm{min.\ irr.\ } } \chi_t(\pi), \qquad \textrm{with} \quad a:=2c_{\systemdispersion}(\ga_1)+ \int_{\bbR^+} \d w \, h(w)   \nonumber  \\
 &  \leq &   b_d(2\ga) \left( \int_{\bbR^+} \d w \, h(w) \e^{-w(\Re{z}-a)}\right)   F\left(\int_{\bbR^+} \d y  \, \int_{\bbR^+} \d w \,  h(y+w)  \e^{-w (\Re{z}-a) } \right). \nonumber
\eaq
where  $F(x):=\frac{x}{1-x}$, provided that $\str x \str <1$.
To prove the first inequality above, we use  \eqref{eq 1 of proof excitations} and \eqref{eq: sum over irreducible}, and, for the second inequality, we use  \eqref{eq: calculation excitation pairings} and  sum the geometric series.

 Statement 2) of Theorem  \ref{thm: polymer model} now follows by fixing the constants and using  the exponential decay of $\hat\psi$. For example, choose $\ga_1, \ga$ such that
\beq
2 c_{\systemdispersion}(\ga_1)  \leq \frac{1}{4} g_\res, \qquad       \ga: = \frac{1}{2} \ga_1 
\eeq
and $\delta'_2$ small enough such that for $\str\la\str \leq \delta' _2$
\beq \int_{\bbR^+}  \d w \, h(w) \leq \frac{1}{4} g_\res, \qquad  \int_{\bbR^+} \d y \int_{\bbR^+} \d w \,  h(y+w)  \e^{-w (-\frac{1}{4} g_\res-a) }\leq 1.   \eeq
Then \eqref{eq: bound rex } is satisfied with $\delta'_1:=  \frac{1}{4} \ga_1$, $g':= \frac{1}{4} g_\res$ and $\delta'_2$ as determined above.

\subsubsection{Proof of equation \eqref{eq: expansion for resolvent} in Statement 1) of Theorem \ref{thm: polymer model} } \label{sec: proof of resolvent formula}

To simplify the following calculations, we abbreviate 
\beq\caR^{\mathrm{irr}}_\la(z):=\caR^{\mathrm{ex}}_\la (z)+\la^2 \caL (z) , \qquad   \caR_\sys(z):= (z- \ad( \i H_\sys))^{-1}. \eeq
By the self-adjointness of  $\ad(H_\sys)$, one has that $\norm  \caR_\sys(z) \norm < \str \Re z\str^{-1}$. 
We choose $\la$ and $z$ such that   $\Re z > 0$   and  $ \norm   \caR^{\mathrm{irr}}_\la(z)    \caR_\sys(z) \norm \leq       \str \Re z\str^{-1}  \norm   \caR^{\mathrm{irr}}_\la(z)    \norm      < 1$. Then
\baq
 \caR_\la(z) &:= & \int_{\bbR^+}  \d t \, \e^{-tz}  \caZ^\la_t  \nonumber\\
   &= & \sum_{n \in \bbZ^+}                    \caR_\sys(z)  \left( \caR^{\mathrm{irr}}_\la(z)    \caR_\sys(z)  \right)^n \nonumber\\
      &= &         \caR_{\sys}(z)  \left(1-    \caR^{\mathrm{irr}}_\la(z)    \caR_\sys(z)  \right)^{-1}    \nonumber\\
      &= &         \left(z- \ad(\i H_\sys)  -  \caR^{\mathrm{irr}}_\la(z)   \right)^{-1}  \nonumber\\
       &= &        \left(z-\ad( \i H_\sys)  - \la^2 \caL (z)  - \caR^{\mathrm{ex}}_\la (z)  \right)^{-1} ,
\eaq
where the second equality follows by Laplace transforming \eqref{eq: z as polymer model}, and the third equality represents  the sum of a geometric series.  Hence, Statement 1) of Theorem \ref{thm: polymer model} is proven.

\section{Proof of Theorem \ref{thm: main technical improved}}  \label{sec: proof of main technical}

In this Section we prove Theorem \ref{thm: main technical improved}.  Our reasoning is based on a standard application of analytic perturbation theory and the inverse Laplace transform.

We abbreviate
\beq
A(z,\la,\ka) := \left( \ad(\i H_\sys)+\la^2\caL(z)+ \caR^{\mathrm{ex}}(z) \right)_{\la^2 \ka}- \la^2 M^{\ka}
\eeq
and we define 
\beq
\mathbf{G}:= \left\{ (z,\la,\ka) \in \bbC\times \bbR \times \bbC^d \,  \Big\str \,      \Re z > -g' ,  \str\ka\str < \delta_{\mathrm{kin}},  \str\la\str < \min( \delta'_2, \sqrt{\frac{\delta'_1}{\delta_{\mathrm{kin}}}} )     \right\},
\eeq
where $g',\delta'_1,\delta'_2$ are as described in Theorem \ref{thm: polymer model} and $\delta_{\mathrm{kin}}$ is as described in Theorem \ref{thm: main technical wc}.
Theorem \ref{thm: polymer model}  implies that, on the domain $\mathbf{G}$, the function $\la^2 M^{\ka}+ A(z,\la,\ka)$ is analytic  in the variables $(z,\ka)$   and, for $\Re z$ large enough, 
\beq \label{def: R in fibers}
 ( \caR_\la(z) )_{\la^2 \ka}= (z-\la^2 M^{\ka}- A(z,\la,\ka))^{-1}.
\eeq
We may extend the (operator-valued) function  $ z \mapsto ( \caR_\la(z) )_{\la^2 \ka}$ into the region $\Re z > -g'$.  This will be useful, because, at the end of this section, we calculate the reduced evolution $(\caZ_t^\la)_{\la^2 \ka}$ from the inverse Laplace transform of $( \caR_\la(z) )_{\la^2 \ka}$.
From \eqref{def: R in fibers} we see that any singular point of the function $ z\mapsto ( \caR_\la(z) )_{\la^2 \ka}$  must satisfy 
\beq \label{z in spectrum}
z \in  \sp ( \la^2 M^{\ka}+ A( z,\la,\ka)).
\eeq
Recall that by Theorem \ref{thm: main technical wc}, $M^{\ka}$ has a simple isolated eigenvalue $f_{\mathrm{kin}}(\ka)$, and let $\Om \subset \bbC$ be as defined in \eqref{def: omega and gap}, i.e., 
\beq
  \Om :=  \mathop{\cup}\limits_{\str\ka\str <  \delta_{\mathrm{kin}}  } \left(\sp M^\ka \setminus \{ f_{\mathrm{kin}}(\ka)\} \right).
\eeq

The following two lemmas describe the singularities of  $(\caZ_t^\la)_{\la^2 \ka}$.
\begin{lemma}\label{lem: difference of sets}
There is a constant $c_1$ and a function $c(\la)$ with $c(\la) \searrow 0$, as $\la \searrow 0$, such that, for any $z$ satisfying \eqref{z in spectrum}, one of the following two statements holds
\beq
\mathrm{dist}( z, \la^2 \Om  ) \leq \la^2 c(\la), \qquad \textrm{\emph{or}}  \qquad  \mathrm{dist}( z, \la^2 f_{\mathrm{kin}}(\ka) )\leq c_1 \la^4.
\eeq 
\end{lemma}
\begin{proof}
From  Theorem \ref{thm: polymer model}, we infer that
\beq\label{bounds on A}
\norm A(z,\la,\ka) \norm  = \la^2 \norm (\caL(z)-\caL(0))_{\la^2 \ka} \norm +O(\la^4)+ O((\la^2 \ka)^2)  \qquad \textrm{as} \, \la \searrow 0,   \la^2 \ka  \searrow 0,
\eeq
with $(\caL(z  )-\caL(0))_{\la^2 \ka} $  bounded and analytic in $(z,\ka)$ on $\mathbf{G}$.
Since $M^{\ka}$ is bounded,  there is a constant $r(m)>0$, for all $m>0$, such that
\beq
 \mathop{\sup}\limits_{z \in \bbC,\, \mathrm{dist}( z, \sp M^\ka  ) \geq r(m) }   \norm (z-M^\ka )^{-1}\norm  \leq    m  .
\eeq
Choose $m^{-1}:= \sup_{(z,\la,\ka) \in \mathbf{G}}  \lakl  \norm A(z,\la,\ka) \norm $ (by \eqref{bounds on A}, $m^{-1}=O(\la^0)$).  Using the Neumann series for $(z- \la^2 M^\ka-A(z,\la,\ka))^{-1} $, it follows that, if  $\mathrm{dist}( z, \la^2\sp M^\ka  ) \geq \la^2 r(m)$, then $z$ cannot satisfy \eqref{z in spectrum}. 

If, however, $\mathrm{dist}( z, \la^2\sp M^\ka  ) \leq \la^2 r(m)$, then $\norm A(z,\la,\ka) \norm= O(\la^4)$, as $\la \searrow 0$; (this follows from \eqref{bounds on A} and the analyticity of $\caL(z)$).   The claim now follows from analytic perturbation theory, using that $ \la^2 f_{\mathrm{kin}}(\ka)$ is an isolated simple eigenvalue.
\end{proof}

\begin{lemma}\label{lem: uniqueness} For sufficiently small $\str \la \str$, there is a 
unique  $z=:\tilde z$  at a distance  $O(\la^4)$ from $\la^2 f_{\mathrm{kin}}(\ka)$ satisfying  \eqref{z in spectrum}.  Let $P^{\la,\ka}$ be the residue of $( z-\la^2 M^{\ka}- A(z,\la,\ka))^{-1}$ at $z=\tilde z$. It follows that $P^{\la,\ka}$ is a rank one-operator and 
\beq \label{ex: P close to projection}
\norm P^{\la,\ka} -   P_{\mathrm{kin}}^{\ka} \norm  = O(\la^2)
\eeq
 with $P_{\mathrm{kin}}^{\ka}$ the one-dimensional spectral projection of $ M^{\ka}$ corresponding to the isolated simple eigenvalue $f_{\mathrm{kin}}(\ka)$, as in Theorem \ref{thm: main technical wc}.
\end{lemma}
\begin{proof} 
By analytic perturbation theory, the operator $\la^2 M^\ka+ A(z,\la,\ka)$ has at most one eigenvalue at a distance  $O(\la^4)$ of $f_{\mathrm{kin}}(\ka)$. This means that   \eqref{z in spectrum} has at most one solution at a distance  $O(\la^4)$ of $f_{\mathrm{kin}}(\ka)$.
We now prove that there is at least one  solution.  Indeed, if no such solution existed,  we could choose a contour 
\beq
\mathbf{C}_{\ka,a}= \{z \in \bbC \,  \str \,   \str z-f_{\mathrm{kin}}(\ka)\str= a \}, \qquad  a>0,
\eeq
 with $a$ small enough such that $\mathbf{C}_{\ka,a}$ stays away from  $\Om$.  We  then calculate 
\baq
 2\pi \i( P_{\mathrm{kin}}^\ka -0)& =& \int_{\la^2 \mathbf{C}_{\ka,a}}  \d z    ( z-\la^2 M^{\ka})^{-1}  - \int_{\la^2 \mathbf{C}_{\ka,a}}   \d z  ( z-\la^2 M^{\ka}- A(z,\la,\ka))^{-1}  \nonumber \\
  &=&  \int_{\la^2 \mathbf{C}_{\ka,a}}  \d z    ( z-\la^2 M^{\ka})^{-1}  \left( 1-  (  1-    A(z,\la,\ka)( z-\la^2 M^{\ka})^{-1} )^{-1}
     \right)      \nonumber \\
       &\leq&  (2 \pi a) \,  b(a,\ka)  \,    \left( 1-  \frac{1}{  1-    b(a,\ka) O(\la^2) } \right) ,   \label{eq: at least one pole}  
\eaq
where
\[
 b(a,\ka):= \sup_{z\in \mathbf{C}_{\ka,a}} \norm (z-M^\ka)^{-1} \norm,
\]
  and, here and in what follows, the contour integrals are meant to be oriented clockwise.
Since the last line of \eqref{eq: at least one pole} is of order $\la^2$, we arrive at a contradiction to the fact that $P_{\mathrm{kin}}^\ka \neq 0$.

The claim about the residue is most easily seen in an abstract setting:
Let $F(z)$ be a Banach-space valued analytic function in some open domain containing $0$, and such that  $0 \in \sp F(0)$ is an isolated eigenvalue.
We have hence the Taylor expansion
\beq
F(z) = \mathop{\sum}_{n \geq 0}  \frac{z^{n}}{n!}  F_{n} , \qquad F_{n} := F^{(n)}(0), \qquad  0 \in \sp F_0
\eeq
If $\norm F_1-1 \norm$ is small enough,  then also $F_1^{-1} F_0$  has $0$ as an isolated eigenvalue. We denote the corresponding spectral projection by $1_0(F_1^{-1} F_0)$ and we calculate
\beq
\mathrm{Res} (F(z)^{-1})  =\mathrm{Res}   (F_0 + z F_1)^{-1}  =  \left(  \mathrm{Res} ( F_1^{-1} F_0 + z )^{-1}  \right)   F_1^{-1} = 1_0(F_1^{-1} F_0)    F_1^{-1}.
\eeq
The last expression is clearly a rank-one operator.  In the case at hand, $F_1^{-1}=1+O(\la^2)$, as $\la \searrow 0$, which yields \eqref{ex: P close to projection}.
\end{proof}
We set $f(\la,\ka):=\tilde z$ and we define  $P^{\la,\ka}$ as the residue of $( z-\la^2 M^{\ka}- A(z,\la,\ka))^{-1}$ at $z=\tilde z$.  It is clear that $f(\la,\ka)$ and $P^{\la,\ka}$ enjoy the analyticity properties claimed in Theorem \ref{thm: main technical improved}.

\begin{figure}[ht!] 
\vspace{0.5cm}
\begin{center}
\psfrag{heightl}{$l$}
\psfrag{realz}{$\Re z$}
\psfrag{imz}{$\i \Im z$}
\psfrag{contourgamma}{$\Ga$}
\psfrag{cgpriem}{$\Ga'$}
\psfrag{contourc}{$\la^2\mathbf{C}_{\ka,a}$}
\psfrag{contourcpriem}{$\la^2\mathbf{C}'$}
\psfrag{gapg}{$\la^2 g$}
\psfrag{gapgres}{$g'-\ep$}
\psfrag{osze}{$O(\la^2)$}
\includegraphics[width = 10cm]{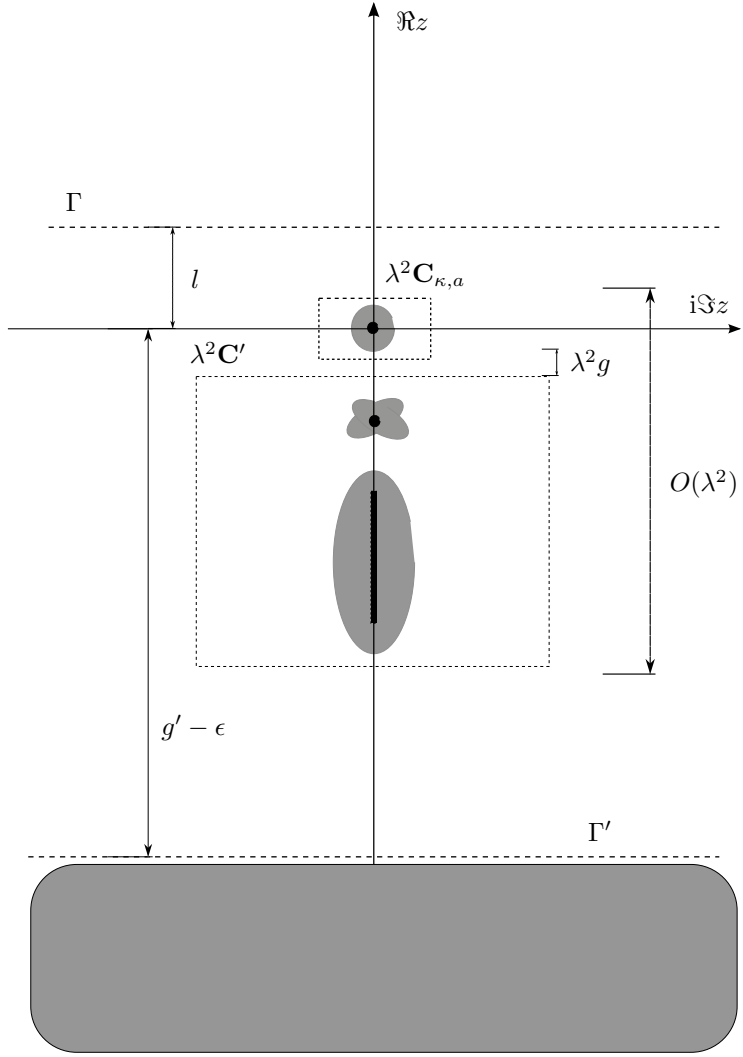}
\end{center}
\caption{ \label{fig: spectral} \footnotesize{The (rotated) complex plane. The black dots and thick black line indicate the spectrum of $\la^2M^{0}$: The upper dot is the eigenvalue $0$ and the thick vertical line is the continuous spectrum. In the  picture, we have drawn only one other eigenvalue, but, in general, there can be more than one (or none) further eigenvalues.  The function $\la^2 M^\ka+ A(z,\la,\ka)$ is  analytic above the lowest gray (rectangular) region. The other gray regions contain the singularities of the function $ ( \caR_\la(z) )_{\la^2 \ka}$ for $(z,\la,\ka) \in \mathbf{G}$.  The integration contours $\Ga, \Ga'$ and $\la^2\mathbf{C}_{\ka,a}, \la^2\mathbf{C}'$ are drawn in dashed lines. In this picture, the contour $\la^2\mathbf{C}_{\ka,a}$ encircles $\la^2f(\la,\ka)$, for all $(\la,\ka)$, (i.e., such that $(z,\la,\ka) \in \mathbf{G}$), which can be achieved by choosing $a$ large enough.}
} 
\end{figure}

We define the horizontal contours
\beq
\Ga: = \{ z \in \bbC \, \str z=l+\i \bbR \}, \qquad   \Ga' : = \{z \in \bbC \,  \str z=-(g'-\ep)+\i \bbR \},
\eeq
with $l$ large enough such that all singular points of  $ z\mapsto ( \caR_\la(z) )_{\la^2 \ka}$ lie below  $\Ga$, and $\ep>0$ small enough such that all singular points with $\Re z >-g'$  lie above $\Ga'$ (the notions 'below' and 'above' are meant as in Figure \ref{fig: spectral}).   These contours are oriented from left to right.
By Theorem \ref{thm: main technical wc}, we can construct  a contour $\mathbf{C}'$  which encircles $\Om$ and such that $ f_{\mathrm{kin}}(\ka)$ is separated by a gap $g$ from this contour:
 \beq g:=\inf_{\str \ka \str \leq   \delta_{\mathrm{kin}} }  \Re f_{\mathrm{kin}}(\ka)-   \sup \Re\mathbf{C}'  >0. \label{def: g'}
 \eeq
%

By performing an  inverse Laplace transform we find that
\beq
({\caZ_t^\la})_{\la^2\ka} = \frac{1}{2\pi \i} \mathop{\int}\limits_{\Ga }\d z \, \e^{tz} ( z-\la^2 M^{\ka}- A(z,\la,\ka))^{-1}.
\eeq
For $\la$ small enough, Lemma \ref{lem: difference of sets} ensures that one can deform contours and obtain
\beq\label{deformed contours}
\mathop{\int}\limits_{\Ga }= \mathop{\int}\limits_{\la^2\mathbf{C}_{a,\ka} }+\mathop{\int}\limits_{\la^2\mathbf{C}'  }+\mathop{\int}\limits_{\Ga' }.
\eeq
The first term on the RHS of \eqref{deformed contours} equals $\e^{ t \la^2  f(\ka,\la) }P^{\la,\ka}$; this follows from Lemma \ref{lem: uniqueness}. The second term is dominated by 
\beq
        \e^{\la^2 t (\sup( \Re\mathbf{C}')) }    \mathop{\int}\limits_{\la^2\mathbf{C}'  }  \frac{\d \str z\str }{2 \pi}    \,  \norm  ( z-\la^2 M^{\ka})^{-1} \norm \norm  \left( 1-  (  1-    A(z,\la,\ka)( z-\la^2 M^{\ka})^{-1} \right)^{-1}\norm.
\eeq
By the choice of $\mathbf{C}'_{\la}$ and the bound \eqref{bounds on A}, the integral on the RHS is bounded by a constant, for $\la$ small enough.

The third term of the RHS of \eqref{deformed contours}  is split as 
\baq
 \mathop{\int}\limits_{\Ga' }  \d z \, \e^{tz} ( z-\la^2 M^{\ka}- A(z,\la,\ka))^{-1} &=&    \mathop{\int}\limits_{\Ga' }  \d z \, \e^{tz} ( z-\la^2 M^{\ka})^{-1}  \\
 &+& \mathop{\int}\limits_{\Ga' }  \d z \, \e^{tz}    ( z-\la^2 M^{\ka})^{-1}  A(z,\la,\ka)( z-\la^2 M^{\ka}- A(z,\la,\ka))^{-1}.  \nonumber  \eaq
The first integral can be closed in the lower half-plane and equals $0$, the  second integral has an integrand of order $z^{-2}$ for large $z$, and hence its contribution is bounded by a constant times $e^{- t (g'-\ep)}$.

It follows that the crucial estimate \eqref{eq: return to equilibrium off fibers} holds with  $\delta_1:= \delta_{\mathrm{kin}} $ and $g$ as in \eqref{def: g'}.

 \section*{APPENDIX A}  
 \label{app: reservoirs}
 
\renewcommand{\theequation}{A-\arabic{equation}}
  \setcounter{equation}{0}  

Here we consider the \emph{effective} structure factor, which, in Section \ref{sec: reservoir},  has been introduced as the Fourier transform of the reservoir correlation function.

We use the  spectral theorem to represent the positive  operator $\bosondispersion$ as  multiplication by $\xi \in \bbR^+$. There are Hilbert spaces $\frh_\xi$ for $\xi \in \bbR^+$ such that 
$
\frh=  \int_{ \oplus \bbR^+} \,  \d \xi \frh_\xi,
$
and for all $\varphi \in \frh$, there are $\varphi_\xi \in \frh_\xi$ such that
\beq
  \varphi= \int_{ \oplus \bbR^+} \,  \d \xi \varphi_\xi,      \qquad   \bosondispersion \varphi =  \int_{ \oplus \bbR^+} \,  \d \xi  \, \xi \,  \varphi_\xi.
\eeq

The structure factor $\phi \in \frh$ has been introduced in Section \ref{sec: reservoir}.  We construct an effective form factor $\phi^\be$ as an element of $\frh\oplus \frh$.
We choose $\frh_{-\xi}$ to be isomorphic to $ \frh_{\xi}$, and 
we define $\phi^\be =  \int_{ \oplus \bbR}  \phi^\be_\xi $ as an element of $  \frh \oplus \frh \sim \int_{ \oplus \bbR}  \frh_\xi$ by setting 
\beq \label{def: psi}
\phi^\be_\xi :=    \left\{ \begin{array}{ll}   \frac{1  }{\sqrt{\e^{\be \xi}-1}}    \, \phi_\xi,     &    \xi>0, \\ 
\frac{1  }{\sqrt{1-\e^{\be \xi}}}    \, \phi_{-\xi},      &    \xi<0.  \end{array} \right.
\eeq
The function $\phi^\be$ plays the role of the form factor if one constructs the positive-temperature dynamical system.  We just note that 
\beq \label{eq: psi as square}
\psi(\xi)= \norm \phi^\be_\xi \norm^2_{\frh_\xi}.
\eeq

Assume that the on-site one-particle space  is given by $\frh= L^2(\bbR^d)$, and the one-particle Hamiltonian acts by multiplication  with a function $\xi(r)$, where $r :=\str q \str$, for $q \in \bbR^d$.  We also assume that $r \mapsto \xi(r)$ is differentiable and monotonically  increasing. Hence we can  define the inverse function $\xi \mapsto r(\xi)$. The form factor $\phi \in L^2(\bbR^d)$ is taken to be spherically symmetric, $ \phi(q) \equiv\phi(r)$. Then the Hilbert spaces $\frh_\xi$ are naturally identified with $L^2(\bbS^{d-1})$,
and 
\beq \label{eq: effective form factor for example}
\phi^\be_\xi =      r(\str\xi\str) ^{\frac{d-1}{2}}   \left( \frac{\partial r(\str \xi\str)}{\partial \str \xi \str} \right)^{-1/2} \,  1_{\bbS^{d-1}} \,      \left\{ \begin{array}{ll}   (\e^{\be \xi}-1)^{-1/2} \,  \phi ( r(\xi)  ),     & \qquad   \xi >0, \\ [2mm]
(1-\e^{\be \xi})^{-1/2} \,   \overline{\phi ( r(-\xi) )},   & \qquad    \xi <0,  \end{array} \right.
\eeq
where $1_{\bbS^{d-1}} \in L^2(\bbS^{d-1}) $ is the constant function on $\bbS^{d-1}$ with  $\norm 1_{\bbS^{d-1}} \norm=1$.

Next, we return to Assumption  \ref{ass: exponential decay}. By  properties of the Fourier transform, e.g.\ Th.\ IX.14  of \cite{reedsimon2}, this assumption is  equivalent to the assumption that $\psi$ extends to an analytic function in the strip $\str \Im \xi \str < g_\res$, and 
\beq \label{cond: hardy class}
\mathop{ \sup}\limits_{ - g_{\res} < y< g_{\res}} \int_{\bbR} \d x \,    \str\psi(x+\i y)\str < \infty.
\eeq
Starting from  expression \eqref{eq: effective form factor for example}, one can check condition \eqref{cond: hardy class}  in concrete examples. 
E.g., for a relativistic dispersion law,  $\xi(r )=r$,  \eqref{cond: hardy class} is satisfied whenever 
\beq
\mathop{ \sup}\limits_{ - g_{\res} < y< g_{\res}} \int_{\bbR}   \d x \, \str x+\i y\str^{d-2} \str \phi(x+\i y) \str^2     < \infty.
\eeq


\bibliographystyle{plain}
\bibliography{mylibrary08}

\end{document}